\documentclass[11pt]{article}
\usepackage{graphicx,amsmath,amsfonts,amscd,amssymb,bm,cite,epsfig,epsf,url,color,algorithm,algorithmic}
\usepackage{fullpage} \usepackage[small,bf]{caption}
\usepackage{subfigure}
\newtheorem{theorem}{Theorem}[section]
\newtheorem{lemma}[theorem]{Lemma}
\newtheorem{corollary}[theorem]{Corollary}

\newtheorem{definition}[theorem]{Definition}

\newcommand{\R}{\mathbb{R}}

\newcommand{\<}{\langle}
\renewcommand{\>}{\rangle}

\newcommand{\sgn}{\textrm{sgn}}
\renewcommand{\P}{\operatorname{\mathbb{P}}}
\newcommand{\E}{\operatorname{\mathbb{E}}}

\newcommand{\PO}{\mathcal{P}_{\Omega}} 
\newcommand{\PT}{\mathcal{P}_T}
\newcommand{\PTp}{\mathcal{P}_{T^\perp}} 
\newcommand{\POc}{\mathcal{P}_{\Omega^\perp}} 
\newcommand{\POj}{\mathcal{P}_{\Omega_j}}
\newcommand{\POzero}{\mathcal{P}_{\Omega_0}}
\newcommand{\POzerop}{\mathcal{P}_{\Omega_0^\perp}}

\newcommand{\POobs}{\mathcal{P}_{\Omega_{\text{obs}}}}
\newcommand{\POobsp}{\mathcal{P}_{\Omega^\perp_{\text{obs}}}}
\newcommand{\Obs}{\Omega_{\text{obs}}}

\newcommand{\rank}{\operatorname{rank}}

\newcommand{\trace}{\operatorname{trace}}

\newcommand{\OpId}{\mathcal{I}}

\numberwithin{equation}{section}

\def \endprf{\hfill {\vrule height6pt width6pt depth0pt}\medskip}
\newenvironment{proof}{\noindent {\bf Proof} }{\endprf\par}

\newcommand{\name}{Principal Component Pursuit}

\title{Robust Principal Component Analysis?}

\author{Emmanuel J. Cand\`es$^{1,2}$, Xiaodong Li$^{2}$, Yi Ma$^{3,4}$, and John Wright$^{4}$\\
  \vspace{-.1cm}\\
  $^1$ Department of Statistics, Stanford University, Stanford, CA 94305\\
  \vspace{-.3cm}\\
  $^2$ Department of Mathematics, Stanford University, Stanford,
  CA 94305\\
 \vspace{-.3cm}\\
  $^{3,4}$ Electrical and Computer Engineering, UIUC, Urbana, IL 61801\\
 \vspace{-.3cm}\\
  $^4$ Microsoft Research Asia, Beijing, China
}

\date{December 17, 2009}

\begin{document}

\maketitle

\vspace{-0.3in}

\begin{abstract}
  This paper is about a curious phenomenon. Suppose we have a data
  matrix, which is the superposition of a low-rank component and a
  sparse component. Can we recover each component individually? We
  prove that under some suitable assumptions, it is possible to
  recover both the low-rank and the sparse components {\em exactly} by
  solving a very convenient convex program called {\em \name}; among
  all feasible decompositions, simply minimize a weighted combination
  of the nuclear norm and of the $\ell_1$ norm.  This suggests the
  possibility of a principled approach to robust principal component
  analysis since our methodology and results assert that one can
  recover the principal components of a data matrix even though a
  positive fraction of its entries are arbitrarily corrupted. This
  extends to the situation where a fraction of the entries are missing
  as well.  We discuss an algorithm for solving this optimization
  problem, and present applications in the area of video surveillance,
  where our methodology allows for the detection of objects in a
  cluttered background, and in the area of face recognition, where it
  offers a principled way of removing shadows and specularities in
  images of faces.
\end{abstract}


\maketitle

\section{Introduction}
\label{sec:intro}

\subsection{Motivation}
\label{sec:motivation}

Suppose we are given a large data matrix $M$, and know
that it may be decomposed as
\[
M = L_0 + S_0,
\]
where $L_0$ has low-rank and $S_0$ is sparse; here, both components
are of arbitrary magnitude.  We do not know the low-dimensional column
and row space of $L_0$, not even their dimension. Similarly, we do not
know the locations of the nonzero entries of $S_0$, not even how many
there are.  Can we hope to recover the low-rank and sparse components
both accurately (perhaps even exactly) and efficiently?

A {\em provably correct} and {\em scalable} solution to the above
problem would presumably have an impact on today's data-intensive
scientific discovery.\footnote{Data-intensive computing is advocated
  by Jim Gray as the fourth paradigm for scientific discovery
  \cite{Tony-Hey}.}  The recent explosion of massive amounts of
high-dimensional data in science, engineering, and society presents a
challenge as well as an opportunity to many areas such as image,
video, multimedia processing, web relevancy data analysis, search,
biomedical imaging and bioinformatics. In such application domains,
data now routinely lie in thousands or even billions of dimensions,
with a number of samples sometimes of the same order of magnitude.

To alleviate the curse of dimensionality and scale,\footnote{We refer
  to either the complexity of algorithms that increases drastically as
  dimension increases, or to their performance that decreases sharply
  when scale goes up.} we must leverage on the fact that such data
have low intrinsic dimensionality, e.g.~that they lie on some
low-dimensional subspace \cite{Eckart1936-Psychometrika}, are sparse
in some basis \cite{ChenS2001-SIAM}, or lie on some low-dimensional
manifold \cite{Tenenbaum2000-Science,Belkin2003-NC}. Perhaps the
simplest and most useful assumption is that the data all lie near some
low-dimensional subspace. More precisely, this says that if we stack
all the data points as column vectors of a matrix $M$, the matrix
should have (approximately) low-rank: mathematically,
\[
M = L_0 + N_0,
\]
where $L_0$ has low-rank and $N_0$ is a small perturbation matrix.
Classical {\em Principal Component Analysis} (PCA)
\cite{Hotelling,Eckart1936-Psychometrika,Jolliffe1986} seeks the best
(in an $\ell^2$ sense) rank-$k$ estimate of $L_0$ by solving
\[
  \begin{array}{ll}
    \text{minimize}   & \quad \|M - L\|\\
    \text{subject to} & \quad \text{rank}(L) \le k. 
  \end{array}
\] 
(Throughout the paper, $\|M\|$ denotes the $2$-norm; that is, the
largest singular value of $M$.) This problem can be efficiently solved
via the singular value decomposition (SVD) and enjoys a number of
optimality properties when the noise $N_0$ is small and
i.i.d. Gaussian.

\paragraph{Robust PCA.} PCA is arguably the most widely used
statistical tool for data analysis and dimensionality reduction
today. However, its brittleness with respect to {\em grossly}
corrupted observations often puts its validity in jeopardy -- a single
grossly corrupted entry in $M$ could render the estimated $\hat{L}$
arbitrarily far from the true $L_0$. Unfortunately, gross errors are
now ubiquitous in modern applications such as image processing, web
data analysis, and bioinformatics, where some measurements may be
arbitrarily corrupted (due to occlusions, malicious tampering, or
sensor failures) or simply irrelevant to the low-dimensional structure
we seek to identify.  A number of natural approaches to robustifying
PCA have been explored and proposed in the literature over several
decades. The representative approaches include influence function
techniques \cite{Huber,DeLaTorre2003-IJCV}, multivariate trimming
\cite{Gnanadesikan1972-Biometrics}, alternating minimization
\cite{Ke2005-CVPR}, and random sampling techniques
\cite{Fischler1981-ACM}. Unfortunately, none of these existing
approaches yields a polynomial-time algorithm with strong performance
guarantees under broad conditions\footnote{Random sampling approaches
  guarantee near-optimal estimates, but have complexity exponential in
  the rank of the matrix $L_0$. Trimming algorithms have comparatively
  lower computational complexity, but guarantee only locally optimal
  solutions.}. The new problem we consider here can be considered as
an idealized version of {\em Robust PCA}, in which we aim to recover a
{\em low-rank} matrix $L_0$ from highly corrupted measurements $M =
L_0 + S_0$. Unlike the small noise term $N_0$ in classical PCA, the
entries in $S_0$ can have arbitrarily large magnitude, and their
support is assumed to be {\em sparse} but unknown\footnote{The unknown
  support of the errors makes the problem more difficult than the
  matrix completion problem that has been recently much studied.}.

\paragraph{Applications.} There are many important applications in
which the data under study can naturally be modeled as a
low-rank plus a sparse contribution. All the statistical applications,
in which robust principal components are sought, of course fit our
model. Below, we give examples inspired by contemporary challenges in
computer science, and note that depending on the applications, either
the low-rank component or the sparse component could be the object of
interest:
\begin{itemize}
\item {\em Video Surveillance.} Given a sequence of surveillance video
  frames, we often need to identify activities that stand out from the
  background. If we stack the video frames as columns of a matrix $M$,
  then the low-rank component $L_0$ naturally corresponds to the
  stationary background and the sparse component $S_0$ captures the
  moving objects in the foreground. However, each image frame has
  thousands or tens of thousands of pixels, and each video fragment
  contains hundreds or thousands of frames. It would be impossible to
  decompose $M$ in such a way unless we have a truly scalable solution
  to this problem. In Section \ref{sec:experiments}, we will show the
  results of our algorithm on video decomposition.

\item {\em Face Recognition.} It is well known that images of a
  convex, Lambertian surface under varying illuminations span a
  low-dimensional subspace \cite{Basri2003-PAMI}. This fact has been
  a main reason why low-dimensional models are mostly effective for
  imagery data. In particular, images of a human's face can be
  well-approximated by a low-dimensional subspace. Being able to
  correctly retrieve this subspace is crucial in many applications
  such as face recognition and alignment. However, realistic face
  images often suffer from self-shadowing, specularities, or
  saturations in brightness, which make this a difficult task and
  subsequently compromise the recognition performance. In Section
  \ref{sec:experiments}, we will show how our method is able to
  effectively remove such defects in face images.

\item {\em Latent Semantic Indexing.} Web search engines often need to
  analyze and index the content of an enormous corpus of documents. A
  popular scheme is the {\em Latent Semantic Indexing} (LSI)
  \cite{Dewester1990-JSIS,Papadimitriou}. The basic idea is to gather
  a document-versus-term matrix $M$ whose entries typically encode the
  relevance of a term (or a word) to a document such as the frequency
  it appears in the document (e.g. the TF/IDF). PCA (or SVD) has
  traditionally been used to decompose the matrix as a low-rank part
  plus a residual, which is not necessarily sparse (as we would like).
  If we were able to decompose $M$ as a sum of a low-rank component
  $L_0$ and a sparse component $S_0$, then $L_0$ could capture common
  words used in all the documents while $S_0$ captures the few key
  words that best distinguish each document from others.

\item {\em Ranking and Collaborative Filtering.} The problem of
  anticipating user tastes is gaining increasing importance in online
  commerce and advertisement. Companies now routinely collect user
  rankings for various products, e.g., movies, books, games, or web
  tools, among which the Netflix Prize for movie ranking is the best
  known \cite{Netflix}. The problem is to use incomplete rankings
  provided by the users on some of the products to predict the
  preference of any given user on any of the products. This problem is
  typically cast as a low-rank matrix completion problem. However, as
  the data collection process often lacks control or is sometimes even
  {\it ad hoc} -- a small portion of the available rankings could be
  noisy and even tampered with. The problem is more challenging since
  we need to simultaneously complete the matrix and correct the
  errors.  That is, we need to infer a low-rank matrix $L_0$ from a
  set of incomplete and corrupted entries.  In Section
  \ref{sec:matrix-completion}, we will see how our results can be
  extended to this situation.
\end{itemize}
Similar problems also arise in many other applications such as
graphical model learning, linear system identification, and coherence
decomposition in optical systems, as discussed in
\cite{Venkat-09}. All in all, the new applications we have listed
above require solving the low-rank and sparse decomposition problem
for matrices of extremely high dimension and under much broader
conditions, a goal this paper aims to achieve.


\subsection{A surprising message} 
\label{sec:surprise}

At first sight, the separation problem seems impossible to solve since
the number of unknowns to infer for $L_0$ and $S_0$ is twice as many
as the given measurements in $M \in \R^{n_1\times n_2}$. Furthermore,
it seems even more daunting that we expect to reliably obtain the
low-rank matrix $L_0$ with errors in $S_0$ of arbitrarily large
magnitude.

In this paper, we are going to see that very surprisingly, not only
can this problem be solved, it can be solved by {\em tractable} convex
optimization. Let $\|M\|_* := \sum_i \sigma_i(M)$ denote the nuclear
norm of the matrix $M$, i.e.~the sum of the singular values of $M$,
and let $\|M\|_1 = \sum_{ij} |M_{ij}|$ denote the $\ell_1$-norm of $M$
seen as a long vector in $\R^{n_1 \times n_2}$.  Then we will show
that under rather weak assumptions, the {\em\name}\, (PCP) estimate
solving\footnote{Although the name naturally suggests an emphasis on
  the recovery of the low-rank component, we reiterate that in some
  applications, the sparse component truly is the object of interest.} 
\begin{equation}
\label{eq:sdp}
  \begin{array}{ll}
    \text{minimize}   & \quad \|L\|_* + \lambda \|S\|_1\\
    \text{subject to} & \quad L + S = M 
  \end{array}
\end{equation}
exactly recovers the low-rank $L_0$ and the sparse
$S_0$. Theoretically, this is guaranteed to work even if the rank of
$L_0$ grows almost linearly in the dimension of the matrix, and the
errors in $S_0$ are up to a constant fraction of all
entries. Algorithmically, we will see that the above problem can be
solved by efficient and scalable algorithms, at a cost not so much
higher than the classical PCA. Empirically, our simulations and
experiments suggest this works under surprisingly broad conditions for
many types of real data. In Section \ref{sec:connections}, we will
comment on the similar approach taken in the paper \cite{Venkat-09},
which was released during the preparation of this manuscript.

\subsection{When does separation make sense?}
\label{sec:incoherence}

A normal reaction is that the objectives of this paper cannot be
met. Indeed, there seems to not be enough information to perfectly
disentangle the low-rank and the sparse components. And indeed, there
is some truth to this, since there obviously is an identifiability
issue. For instance, suppose the matrix $M$ is equal to $e_1 e_1^*$
(this matrix has a one in the top left corner and zeros everywhere
else). Then since $M$ is both sparse and low-rank, how can we decide
whether it is low-rank or sparse? To make the problem meaningful, we
need to impose that the low-rank component $L_0$ is not sparse.  In
this paper, we will borrow the general notion of incoherence
introduced in \cite{CR08} for the matrix completion problem; this is
an assumption concerning the singular vectors of the low-rank
component. Write the singular value decomposition of $L_0 \in \R^{n_1
  \times n_2}$ as
\[
L_0 = U \Sigma V^* = \sum_{i = 1}^r \sigma_i u_i v_i^*, 
\]
where $r$ is the rank of the matrix, $\sigma_1, \ldots, \sigma_r$ are
the positive singular values, and $U = [u_1, \ldots, u_r]$, $V = [v_1,
\ldots, v_r]$ are the matrices of left- and right-singular
vectors. Then the incoherence condition with parameter $\mu$ states
that
\begin{equation}
  \label{eq:PU}
  \max_i \|U^* e_i\|^2 \le \frac{\mu r}{n_1}, \quad \max_i \|V^*e_i\|^2 \le
   \frac{\mu r}{n_2}, 
\end{equation}
and 
\begin{equation}
  \label{eq:UV}
  \|UV^*\|_\infty \le  \sqrt{\frac{\mu r}{n_1 n_2}}. 
\end{equation}
Here and below, $\|M\|_\infty = \max_{i,j} |M_{ij}|$, i.e.~is the
$\ell_\infty$ norm of $M$ seen as a long vector. Note that since the
orthogonal projection $P_U$ onto the column space of $U$ is given by
$P_U = UU^*$, \eqref{eq:PU} is equivalent to $\max_i \|P_U e_i\|^2 \le
\mu r/n_1$, and similarly for $P_V$.  As discussed in earlier
references \cite{CR08,CT09,GrossMC}, the incoherence condition asserts
that for small values of $\mu$, the singular vectors are reasonably
spread out -- in other words, not sparse.

Another identifiability issue arises if the sparse matrix has
low-rank.  This will occur if, say, all the nonzero entries of $S$
occur in a column or in a few columns. Suppose for instance, that the
first column of $S_0$ is the opposite of that of $L_0$, and that all
the other columns of $S_0$ vanish. Then it is clear that we would not
be able to recover $L_0$ and $S_0$ by any method whatsoever since $M =
L_0 + S_0$ would have a column space equal to, or included in that of
$L_0$. To avoid such meaningless situations, we will assume that the
sparsity pattern of the sparse component is selected uniformly at
random.

\subsection{Main result}
\label{sec:main}

The surprise is that under these minimal assumptions, the simple PCP
solution perfectly recovers the low-rank and the sparse components,
provided of course that the rank of the low-rank component is not too
large, and that the sparse component is reasonably sparse. Below,
$n_{(1)} = \text{max}(n_1,n_2)$ and $n_{(2)} = \text{min}(n_1,n_2)$.
\begin{theorem}
\label{teo:main}
Suppose $L_0$ is $n \times n$, obeys \eqref{eq:PU}--\eqref{eq:UV}, and
that the support set of $S_0$ is uniformly distributed among all sets
of cardinality $m$. Then there is a numerical constant $c$ such that
with probability at least $1 - c n^{-10}$ (over the choice of support
of $S_0$), \name\, \eqref{eq:sdp} with $\lambda = 1/\sqrt{n}$ is {\em
  exact}, i.e.~$\hat L = L_0$ and $\hat S = S_0$, provided that
\begin{equation}
\label{eq:main}
\rank(L_0) \le \rho_r  n\, \mu^{-1} (\log n)^{-2} \quad \text{ and } 
\quad m \le \rho_s\, n^2.
\end{equation}
Above, $\rho_r$ and $\rho_s$ are positive numerical constants. In the
general rectangular case where $L_0$ is $n_1 \times n_2$, PCP with
$\lambda = 1/\sqrt{n_{(1)}}$ succeeds with probability at least $1 - c
n_{(1)}^{-10}$, provided that $\rank(L_0) \le \rho_r n_{(2)} \,
\mu^{-1} (\log n_{(1)})^{-2}$ and $m \le \rho_s\, n_1 n_2$.
\end{theorem}
In other words, matrices $L_0$ whose singular vectors---or principal
components---are reasonably spread can be recovered with probability
nearly one from arbitrary and completely unknown corruption patterns
(as long as these are randomly distributed). In fact, this works for
large values of the rank, i.e.~on the order of $n/(\log n)^2$ when
$\mu$ is not too large.  We would like to emphasize that the only
`piece of randomness' in our assumptions concerns the locations of the
nonzero entries of $S_0$; everything else is deterministic.  In
particular, all we require about $L_0$ is that its singular vectors
are not spiky. Also, we make no assumption about the magnitudes or
signs of the nonzero entries of $S_0$. To avoid any ambiguity, our
model for $S_0$ is this: take an {\em arbitrary} matrix $S$ and set to
zero its entries on the random set $\Omega^c$; this gives $S_0$.

A rather remarkable fact is that there is no tuning parameter in our
algorithm. Under the assumption of the theorem, minimizing
\[
\|L\|_* + \frac{1}{\sqrt{n_{(1)}}} \|S\|_1, \quad n_{(1)} =
\textrm{max}(n_1,n_2)
\]
always returns the correct answer. This is surprising because one
might have expected that one would have to choose the right scalar
$\lambda$ to balance the two terms in $\|L\|_* + \lambda \|S\|_1$
appropriately (perhaps depending on their relative size). This is,
however, clearly not the case. In this sense, the choice $\lambda =
1/\sqrt{n_{(1)}}$ is universal.  Further, it is not a priori very
clear why $\lambda = 1/\sqrt{n_{(1)}}$ is a correct choice no matter
what $L_0$ and $S_0$ are. It is the mathematical analysis which
reveals the correctness of this value. In fact, the proof of the
theorem gives a whole range of correct values, and we have selected a
sufficiently simple value in that range.

Another comment is that one can obtain results with larger
probabilities of success, i.e.~of the form $1 - O(n^{-\beta})$ (or $1
- O(n_{(1)}^{-\beta})$) for $\beta > 0$ at the expense of reducing the
value of $\rho_r$.

\subsection{Connections with prior work and innovations}
\label{sec:connections}

The last year or two have seen the rapid development of a scientific
literature concerned with the {\em matrix completion} problem
introduced in \cite{CR08}, see also
\cite{CT09,Candes2009-PIEEE,Montanari-IT,GrossQuantum,GrossMC} and the
references therein. In a nutshell, the matrix completion problem is
that of recovering a low-rank matrix from only a small fraction of its
entries, and by extension, from a small number of linear functionals.
Although other methods have been proposed \cite{Montanari-IT}, the
method of choice is to use convex optimization
\cite{CT09,Candes2009-PIEEE,GrossQuantum,GrossMC,Recht2008-SR}: among
all the matrices consistent with the data, simply find that with
minimum nuclear norm. The papers cited above all prove the
mathematical validity of this approach, and our mathematical analysis
borrows ideas from this literature, and especially from those
pioneered in \cite{CR08}. Our methods also much rely on the powerful
ideas and elegant techniques introduced by David Gross in the context
of quantum-state tomography \cite{GrossQuantum,GrossMC}. In
particular, the clever golfing scheme \cite{GrossMC} plays a crucial
role in our analysis, and we introduce two novel modifications to this
scheme.

Despite these similarities, our ideas depart from the literature on
matrix completion on several fronts. First, our results obviously are
of a different nature.  Second, we could think of our separation
problem, and the recovery of the low-rank component, as a matrix
completion problem. Indeed, instead of having a fraction of observed
entries available and the other missing, we have a fraction available,
but do not know which one, while the other is not missing but entirely
corrupted altogether. Although, this is a harder problem, one way to
think of our algorithm is that it simultaneously detects the corrupted
entries, and perfectly fits the low-rank component to the remaining
entries that are deemed reliable. In this sense, our methodology and
results go beyond matrix completion.  Third, we introduce a novel
de-randomization argument that allows us to fix the signs of the
nonzero entries of the sparse component. We believe that this
technique will have many applications. One such application is in the
area of compressive sensing, where assumptions about the randomness of
the signs of a signal are common, and merely made out of convenience
rather than necessity; this is important because assuming independent
signal signs may not make much sense for many practical applications
when the involved signals can all be non-negative (such as images).

We mentioned earlier the related work \cite{Venkat-09}, which also
considers the problem of decomposing a given data matrix into sparse
and low-rank components, and gives sufficient conditions for convex
programming to succeed. These conditions are phrased in terms of two
quantities. The first is the maximum ratio between the $\ell_\infty$
norm and the operator norm, restricted to the subspace generated by
matrices whose row or column spaces agree with those of $L_0$. The
second is the maximum ratio between the operator norm and the
$\ell_\infty$ norm, restricted to the subspace of matrices that vanish
off the support of $S_0$. Chandrasekaran et.\ al.\ show that when the
product of these two quantities is small, then the recovery is exact
for a certain interval of the regularization parameter
\cite{Venkat-09}.

One very appealing aspect of this condition is that it is completely
deterministic: it does not depend on any random model for $L_0$ or
$S_0$. It yields a corollary that can be easily compared to our
result: suppose $n_1 = n_2 = n$ for simplicity, and let $\mu_0$ be the
smallest quantity satisfying (1.2), then correct recovery occurs
whenever
\[ 
\max_{j} \{i : [S_0]_{ij} \neq 0\} \times \sqrt{\mu_0 r/n} < 1/12.
\]
The left-hand side is at least as large as $\rho_s \sqrt{\mu_0 n r}$,
where $\rho_s$ is the fraction of entries of $S_0$ that are
nonzero. Since $\mu_0 \ge 1$ always, this statement only guarantees
recovery if $\rho_s = O((nr)^{-1/2})$; i.e., even when
$\mathrm{rank}(L_0) = O(1)$, only vanishing fractions of the entries
in $S_0$ can be nonzero.

In contrast, our result shows that for incoherent $L_0$, correct
recovery occurs with high probability for $\mathrm{rank}(L_0)$ on the
order of $n/[\mu \log^2 n]$ and a number of nonzero entries in $S_0$
on the order of $n^2$. That is, matrices of large rank can be
recovered from non-vanishing fractions of sparse errors. This
improvement comes at the expense of introducing one piece of
randomness: a uniform model on the error support.\footnote{Notice that
the bound of \cite{Venkat-09} depends only on the support of $S_0$,
and hence can be interpreted as a worst case result with respect to
the signs of $S_0$. In contrast, our result does not randomize over
the signs, but does assume that they are sampled from a fixed sign
pattern. Although we do not pursue it here due to space limitations,
our analysis also yields a result which holds for worst case sign
patterns, and guarantees correct recovery with $\mathrm{rank}(L_0) =
O(1)$, and a sparsity pattern of cardinality $\rho n_1 n_2$ for some $\rho
> 0$.}

Our analysis has one additional advantage, which is of significant
practical importance: it identifies a simple, non-adaptive choice of
the regularization parameter $\lambda$. In contrast, the conditions on
the regularization parameter given by Chandrasekaran et al.~depend on
quantities which in practice are not known a-priori. The experimental
section of \cite{Venkat-09} suggests searching for the correct
$\lambda$ by solving many convex programs. Our result, on the other
hand, demonstrates that the simple choice $\lambda = 1/\sqrt{n}$ works
with high probability for recovering any square incoherent matrix.


\subsection{Implications for matrix completion from grossly corrupted
  data}
\label{sec:matrix-completion}

We have seen that our main result asserts that it is possible to
recover a low-rank matrix even though a significant fraction of its
entries are corrupted. In some applications, however, some of the
entries may be missing as well, and this section addresses this
situation. Let $\PO$ be the orthogonal projection onto the linear
space of matrices supported on $\Omega \subset [n_1] \times [n_2]$,
\[
\PO X = \begin{cases} X_{ij}, & (i,j) \in \Omega,\\
 0, & (i,j) \notin
  \Omega.
\end{cases}
\]
Then imagine we only have available a few entries of $L_0 + S_0$,
which we conveniently write as
\[
Y = \POobs(L_0 + S_0) = \POobs L_0 + S'_0;
\]
that is, we see only those entries $(i,j) \in \Obs \subset [n_1]
\times [n_2]$. 
This models the following problem: we wish to recover $L_0$ but only
see a few entries about $L_0$, and among those a fraction happens to
be corrupted, and we of course do not know which one. As is easily
seen, this is a significant extension of the matrix completion
problem, which seeks to recover $L_0$ from undersampled but otherwise
perfect data $\POobs L_0$.

We propose recovering $L_0$ by solving the following problem:
\begin{equation}
  \label{eq:sdp2}
   \begin{array}{ll}
     \text{minimize}   & \quad \|L\|_* + \lambda \|S\|_1\\
     \text{subject to} & \quad \POobs(L + S) = Y. 
  \end{array}
\end{equation}
In words, among all decompositions matching the available data,
\name\, finds the one that minimizes the weighted combination of the
nuclear norm, and of the $\ell_1$ norm. Our observation is that under
some conditions, this simple approach recovers the low-rank component
exactly. In fact, the techniques developed in this paper establish
this result:
\begin{theorem}
\label{teo:MCRobust}
Suppose $L_0$ is $n \times n$, obeys the conditions
\eqref{eq:PU}--\eqref{eq:UV}, and that $\Obs$ is uniformly distributed
among all sets of cardinality $m$ obeying $m = 0.1 n^2$.  Suppose for
simplicity, that each observed entry is corrupted with probability
$\tau$ independently of the others.  Then there is a numerical
constant $c$ such that with probability at least $1 - c n^{-10}$,
\name\, \eqref{eq:sdp2} with $\lambda = 1/\sqrt{0.1 n}$ is
exact, i.e.~$\hat L = L_0$, provided that
\begin{equation}
  \label{eq:MCrobust}
  \rank(L_0) \le \rho_r  \, n 
  \mu^{-1} (\log n)^{-2}, \quad \text{and} \quad   
  \tau \le \tau_s. 
\end{equation}
Above, $\rho_r$ and $\tau_s$ are positive numerical constants. For
general $n_1 \times n_2$ rectangular matrices, PCP with $\lambda =
1/\sqrt{0.1 n_{(1)}}$ succeeds from $m = 0.1 n_1 n_2$ corrupted
entries with probability at least $1 - c n_{(1)}^{-10}$, provided
that $\rank(L_0) \le \rho_r \, n_{(2)} \mu^{-1} (\log n_{(1)})^{-2}$.
\end{theorem}
In short, perfect recovery from incomplete and corrupted entries is
possible by convex optimization.

On the one hand, this result extends our previous result in the
following way. If all the entries are available, i.e.~$m = n_1 n_2$,
then this is Theorem \ref{teo:main}. On the other hand, it extends
matrix completion results. Indeed, if $\tau = 0$, we have a pure
matrix completion problem from about a fraction of the total number of
entries, and our theorem guarantees perfect recovery as long as $r$
obeys \eqref{eq:MCrobust}, which for large values of $r$, matches the
strongest results available. We remark that the recovery is exact,
however, via a different algorithm. To be sure, in matrix completion
one typically minimizes the nuclear norm $\|L\|_*$ subject to the
constraint $\POobs L = \POobs L_0$. Here, our program would solve
\begin{equation}
   \begin{array}{ll}
     \text{minimize}   & \quad \|L\|_* + \lambda \|S\|_1\\
     \text{subject to} & \quad \POobs(L + S) = \POobs L_0,  
  \end{array}
\end{equation}
and return $\hat L = L_0$, $\hat S = 0$! In this context, Theorem
\ref{teo:MCRobust} proves that matrix completion is stable vis a vis
gross errors.

\paragraph{Remark.} We have stated Theorem \ref{teo:MCRobust} merely
to explain how our ideas can easily be adapted to deal with low-rank
matrix recovery problems from undersampled and possibly grossly
corrupted data. In our statement, we have chosen to see 10\% of the
entries but, naturally, similar results hold for all other positive
fractions provided that they are large enough.  We would like to make
it clear that a more careful study is likely to lead to a stronger
version of Theorem \ref{teo:MCRobust}. In particular, for very low
rank matrices, we expect to see similar results holding with far fewer
observations; that is, in the limit of large matrices, from a
decreasing fraction of entries. In fact, our techniques would already
establish such sharper results but we prefer not to dwell on such
refinements at the moment, and leave this up for future work.

\subsection{Notation}
\label{sec:notation}

We provide a brief summary of the notations used throughout the
paper. We shall use five norms of a matrix. The first three are
functions of the singular values and they are: 1) the operator norm or
$2$-norm denoted by $\|X\|$; 2) the Frobenius norm denoted by
$\|X\|_F$; and 3) the nuclear norm denoted by $\|X\|_*$. The last two
are the $\ell_1$ and $\ell_\infty$ norms of a matrix seen as a long
vector, and are denoted by $\|X\|_1$ and $\|X\|_\infty$ respectively.
The Euclidean inner product between two matrices is defined by the
formula $\<X, Y\> := \trace({X}^* {Y})$, so that $\|X\|_F^2 =
\<X,X\>$.

Further, we will also manipulate linear transformations which act on
the space of matrices, and we will use calligraphic letters for these
operators as in $\PO X$. 
We shall also abuse notation by also letting $\Omega$ be the linear
space of matrices supported on $\Omega$. Then $\POc$ denotes the
projection onto the space of matrices supported on $\Omega^c$ so that
$\OpId = \PO + \POc$, where $\OpId$ is the identity operator. We will
consider a single norm for these, namely, the operator norm (the top
singular value) denoted by $\|\mathcal{A}\|$, which we may want to
think of as $\|\mathcal{A}\| = \sup_{\{\|X\|_F = 1\}} \|\mathcal{A}
X\|_F$; for instance, $\|\PO\| = 1$ whenever $\Omega \neq \emptyset$.

\subsection{Organization of the paper}
\label{sec:organization}

The paper is organized as follows. In Section \ref{sec:architecture},
we provide the key steps in the proof of Theorem \ref{teo:main}. This
proof depends upon on two critical properties of dual certificates,
which are established in the separate Section \ref{sec:lemmas}. The
reason why this is separate is that in a first reading, the reader
might want to jump to Section \ref{sec:experiments}, which presents
applications to video surveillance, and computer vision. Section
\ref{sec:algorithms} introduces algorithmic ideas to find the \name\,
solution when $M$ is of very large scale. We conclude the paper with a
discussion about future research directions in Section
\ref{sec:discussion}. Finally, the proof of Theorem \ref{teo:MCRobust}
is in the Appendix, Section \ref{sec:appendix}, together with those of
intermediate results.

\section{Architecture of the Proof}
\label{sec:architecture}

This section introduces the key steps underlying the proof of our main
result, Theorem \ref{teo:main}. We will prove the result for square
matrices for simplicity, and write $n = n_1 = n_2$. Of course, we
shall indicate where the argument needs to be modified to handle the
general case.  Before we start, it is helpful to review some basic
concepts and introduce additional notation that shall be used
throughout.  For a given scalar $x$, we denote by $\sgn(x)$ the sign
of $x$, which we take to be zero if $x = 0$. By extension, $\sgn(S)$
is the matrix whose entries are the signs of those of $S$.  We recall
that any subgradient of the $\ell_1$ norm at $S_0$ supported on
$\Omega$, is of the form
\[
\sgn(S_0) + F, 
\]
where $F$ vanishes on $\Omega$, i.e.~$\PO F = 0$, and obeys
$\|F\|_\infty \le 1$.

We will also manipulate the set of subgradients of the nuclear
norm. From now on, we will assume that $L_0$ of rank $r$ has the
singular value decomposition $U\Sigma V^*$, where $U, V \in
\R^{n\times r}$ just as in Section \ref{sec:incoherence}. Then any
subgradient of the nuclear norm at $L_0$ is of the form
\[
UV^* + W, 
\]
where $U^* W = 0$, $W V = 0$ and $\|W\| \le 1$. Denote by $T$ the
linear space of matrices 
\begin{equation}
  \label{eq:T}
  T := \{UX^* + YV^*, \, X, Y \in \R^{n \times r}\}, 
\end{equation}
and by $T^\perp$ its orthogonal complement. It is not hard to see that
taken together, $U^* W = 0$ and $W V = 0$ are equivalent to $\PT W =
0$, where $\PT$ is the orthogonal projection onto $T$. Another way to
put this is $\PTp W = W$. In passing, note that for any matrix $M$,
$\PTp M = (I - UU^*) M (I - VV^*)$, where we recognize that $I - UU^*$
is the projection onto the orthogonal complement of the linear space
spanned by the columns of $U$ and likewise for $(I - VV^*)$. A
consequence of this simple observation is that for any matrix $M$,
$\|\PTp M\| \le \|M\|$, a fact that we will use several times in the
sequel. Another consequence is that for any matrix of the form $e_i e_j^*$, 
\[
\|\PTp e_i e_j^*\|_F^2 = \|(I - UU^*) e_i\|^2 \|(I - VV^*) e_j\|^2 \ge
(1-\mu r/n)^2,
\]
where we have assumed $\mu r/n \le 1$. Since $\|\PT e_ie_j^*\|_F^2 +
\|\PTp e_i e_j^*\|_F^2 = 1$, this gives
\begin{equation}
  \label{eq:PTeiej}
  \|\PT e_i e_j^*\|_F \le \sqrt{\frac{2\mu r}{n}}. 
\end{equation}
For rectangular matrices, the estimate is $ \|\PT e_i e_j^*\|_F \le
\sqrt{\frac{2\mu r}{\min(n_1,n_2)}}$. 

Finally, in the sequel we will write that an event holds with high or
large probability whenever it holds with probability at least $1 -
O(n^{-10})$ (with $n_{(1)}$ in place of $n$ for rectangular matrices).

\subsection{An elimination theorem}
\label{sec:elimination}

We begin with a useful definition and an elementary result we
shall use a few times.
\begin{definition}
  We will say that $S'$ is a trimmed version of $S$ if
  $\text{supp}(S') \subset \text{supp}(S)$ and $S'_{ij} = S_{ij}$
  whenever $S'_{ij} \neq 0$. 
\end{definition}
In words, a trimmed version of $S$ is obtained by setting some of the
entries of $S$ to zero.  Having said this, the following intuitive
theorem asserts that if \name\, correctly
recovers the low-rank and sparse components of $M_0 = L_0 + S_0$, it
also correctly recovers the components of a matrix $M'_0 = L_0 + S'_0$
where $S'_0$ is a trimmed version of $S_0$. This is intuitive since
the problem is somehow easier as there are fewer things to recover.
\begin{theorem}
\label{teo:elimination}
Suppose the solution to \eqref{eq:sdp} with input data $M_0 = L_0 +
S_0$ is unique and exact, and consider $M'_0 = L_0 + S'_0$, where $S'_0$
is a trimmed version of $S_0$. Then the solution to \eqref{eq:sdp}
with input $M'_0$ is exact as well.
\end{theorem}
\begin{proof}
  Write $S'_0 = \POzero S_0$ for some $\Omega_0 \subset [n]
  \times [n]$ and let $(\hat{L}, \hat{S})$ be the solution of
  \eqref{eq:sdp} with input $L_0+S'_0$. Then
\[
\|\hat{L}\|_* + \lambda \|\hat{S}\|_1 \le \|L_0\|_* + \lambda
\|\POzero S_0\|_1
\]
and, therefore,
\[
\|\hat{L}\|_* + \lambda \|\hat{S}\|_1 + \lambda \|\POzerop S_0\|_1
\le \|L_0\|_* + \lambda \|S_0\|_1.
\]
Note that $({\hat L}, \hat{S} + \POzerop S_0)$ is feasible
for the problem with input data $L_0 + S_0$, and since $\|\hat{S} +
\POzerop S_0\|_1 \le \|\hat{S}\|_1 +
\|\POzerop S_0\|_1$, we have
\[
\|\hat{L}\|_* + \lambda \|\hat{S} + \POzerop S_0\|_1 \le \|L_0\|_* +
\lambda \|S_0\|_1. 
\]
The right-hand side, however, is the optimal value, and by unicity of
the optimal solution, we must have $\hat{L} = L_0$, and $\hat{S} +
\POzerop S_0 = S_0$ or $\hat{S} = \POzero S_0 = S'_0$. This proves the
claim.
\end{proof} 

\paragraph{The Bernoulli model.} In Theorem \ref{teo:main},
probability is taken with respect to the uniformly random subset
$\Omega = \{(i,j) : S_{ij} \neq 0\}$ of cardinality $m$. In practice,
it is a little more convenient to work with the {\em Bernoulli model}
$\Omega = \{(i,j) : \delta_{ij} = 1\}$, where the $\delta_{ij}$'s are
i.i.d.~variables Bernoulli taking value one with probability $\rho$
and zero with probability $1-\rho$, so that the expected cardinality
of $\Omega$ is $\rho n^2$. From now on, we will write $\Omega \sim
\text{Ber}(\rho)$ as a shorthand for $\Omega$ is sampled from the
Bernoulli model with parameter $\rho$.

Since by Theorem \ref{teo:elimination}, the success of the algorithm
is monotone in $|\Omega|$, any guarantee proved for the Bernoulli
model holds for the uniform model as well, and vice versa, if we allow
for a vanishing shift in $\rho$ around $m/n^2$. The arguments
underlying this equivalence are standard, see \cite{CRT,CT09}, and may
be found in the Appendix for completeness.

\subsection{Derandomization}
\label{sec:derandom}

In Theorem \ref{teo:main}, the values of the nonzero entries of $S_0$
are fixed. It turns out that it is easier to prove the theorem under a
stronger assumption, which assumes that the signs of the nonzero
entries are independent symmetric Bernoulli variables, i.e.~take the
value $\pm 1$ with probability $1/2$ (independently of the choice of
the support set). The convenient theorem below shows that establishing
the result for random signs is sufficient to claim a similar result
for fixed signs.
\begin{theorem}
  \label{teo:derandom}
  Suppose $L_0$ obeys the conditions of Theorem \ref{teo:main} and
  that the locations of the nonzero entries of $S_0$ follow the
  Bernoulli model with parameter $2\rho_s$, and the signs of $S_0$ are
  i.i.d. $\pm 1$ as above (and independent from the locations). Then
  if the PCP solution is exact with high probability, then it is also
  exact with at least the same probability for the model in which the
  signs are fixed and the locations are sampled from the Bernoulli
  model with parameter $\rho_s$.
\end{theorem}
This theorem is convenient because to prove our main result, we only
need to show that it is true in the case where the signs of the sparse
component are random. 
 
\begin{proof}
  Consider the model in which the signs are fixed. In this model, it
  is convenient to think of $S_0$ as $\PO S$, for some fixed matrix
  $S$, where $\Omega$ is sampled from the Bernoulli model with
  parameter $\rho_s$. Therefore, $S_0$ has independent components
  distributed as
 \begin{equation*}
   (S_0)_{ij} = \begin{cases} S_{ij}, &
     \text{w.~p. }
     \rho_s,\\
     0, & \text{w.~p. } 1-\rho_s.
\end{cases} 
\end{equation*}
Consider now a random sign matrix with i.i.d.~entries distributed as
\[
E_{ij} = \begin{cases} 1, & \text{w.~p. }
  \rho_s,\\
  0, & \text{w.~p. } 1-2\rho_s,\\
  -1, & \text{w.~p. } \rho_s,
\end{cases} 
\]
and an ``elimination'' matrix $\Delta$ with entries defined by
\[
\Delta_{ij} =  \begin{cases} 0, & \text{if }
    E_{ij} [\sgn(S)]_{ij} = -1,\\
1, & \text{otherwise}.
\end{cases} 
\]
Note that the entries of $\Delta$ are independent since they are
functions of independent variables.

Consider now $S'_0 = \Delta \circ (|S| \circ E)$, where $\circ$
denotes the Hadamard or componentwise product so that, $[S'_0]_{ij} =
\Delta_{ij} \, (|S_{ij}| E_{ij})$. Then we claim that $S'_0$ and $S_0$
have the same distribution. To see why this is true, it suffices by
independence to check that the marginals match. For $S_{ij} \neq 0$,
we have
\begin{align*}
  \P([S'_0]_{ij} = S_{ij}) & = \P(\Delta_{ij} = 1 \text{ and } E_{ij}
  =
  [\sgn(S)]_{ij})\\
  & = \P( E_{ij} [\sgn(S)]_{ij} \neq -1 \text{ and } E_{ij} =
  [\sgn(S)]_{ij})\\
  & = \P(E_{ij} = [\sgn(S)]_{ij}) = \rho_s,
\end{align*}
which establishes the claim. 

This construction allows to prove the theorem. Indeed, $|S| \circ E$
now obeys the random sign model, and by assumption, PCP recovers $|S|
\circ E$ with high probability. By the elimination theorem, this
program also recovers $S'_0 = \Delta \circ (|S| \circ E)$. Since
$S'_0$ and $S_0$ have the same distribution, the theorem follows.
\end{proof}

\subsection{Dual certificates}
\label{sec:kkt}

We introduce a simple condition for the pair $(L_0,S_0)$ to be the
unique optimal solution to \name. These conditions are stated in terms
of a dual vector, the existence of which certifies optimality. (Recall
that $\Omega$ is the space of matrices with the same support as the
sparse component $S_0$, and that $T$ is the space defined via the the
column and row spaces of the low-rank component $L_0$ \eqref{eq:T}.)
\begin{lemma}
\label{teo:kkt}
Assume that $\|\PO \PT\| < 1$. With the standard notations, $(L_0,
S_0)$ is the unique solution if there is a pair $(W, F)$ obeying
\[
UV^* + W = \lambda(\sgn(S_0) + F),  
\]
with $\PT W = 0$, $\|W\| < 1$, $\PO F = 0$ and $\|F\|_\infty < 1$.
\end{lemma}
Note that the condition $\|\PO \PT\| < 1$ is equivalent to saying that
$\Omega \cap T = \{0\}$.

\begin{proof}
  We consider a feasible perturbation $(L_0 + H, S_0 - H)$ and show
  that the objective increases whenever $H \neq 0$, hence proving that
  $(L_0,S_0)$ is the unique solution. To do this, let $UV^* + W_0$ be
  an arbitrary subgradient of the nuclear norm at $L_0$, and
  $\sgn(S_0) + F_0$ be an arbitrary subgradient of the $\ell_1$-norm
  at $S_0$. By definition of subgradients,
\[
\|L_0 + H\|_* + \lambda \|S_0 - H\|_1 \ge \|L_0\|_* + \lambda
\|S_0\|_1 + \<UV^* + W_0, H\> - \lambda \<\sgn(S_0) + F_0, H\>.
\]
Now pick $W_0$ such that $\<W_0, H\> = \|\PTp H\|_*$ and $F_0$ such
that $\<F_0, H\> = -\|\POc H\|_1$.\footnote{For instance, $F_0 = -
  \sgn(\POc H)$ is such a matrix. Also, by duality between the nuclear
  and the operator norm, there is a matrix obeying $\|W\| = 1$ such
  that $\<W, \PTp H\> = \|\PTp H\|_*$, and we just take $W_0 =
  \PTp(W)$.}  We have
\[
\|L_0 + H\|_* + \lambda \|S_0 - H\|_1 \ge \|L_0\|_* + \lambda
\|S_0\|_1 + \|\PTp H\|_* + \lambda \|\POc H\|_1 + \<UV^*- \lambda
\sgn(S_0), H\>.
\]
By assumption
\[
|\<UV^*- \lambda \sgn(S_0), H\>| \le |\<W,H\>| + \lambda |\<F, H\>| \le \beta
(\|\PTp H\|_* + \lambda \|\POc H\|_1)
\]
for $\beta = \text{max}(\|W\|, \|F\|_\infty) < 1$ and, thus,
\[
\|L_0 + H\|_* + \lambda \|S_0 - H\|_1 \ge \|L_0\|_* + \lambda
\|S_0\|_1 + (1-\beta) \Bigl(\|\PTp H\|_* + \lambda \|\POc H\|_1\Bigr). 
\]
Since by assumption, $\Omega \cap T = \{0\}$, we have $\|\PTp H\|_* +
\lambda \|\POc H\|_1 > 0$ unless $H = 0$. 
\end{proof}

Hence, we see that to prove exact recovery, it is sufficient to produce a
`dual certificate' $W$ obeying
\begin{equation}
\label{eq:dual-certif}
\begin{cases} W \in T^\perp, \\
  \|W\| < 1,\\
  \PO(UV^* + W) = \lambda \sgn(S_0),\\
  \|\POc(UV^* + W)\|_\infty < \lambda.\\
\end{cases}
\end{equation}
Our method, however, will produce with high probability a slightly
different certificate. The idea is to slightly relax the constraint
$\PO(UV^* + W) = \lambda \sgn(S_0)$, a relaxation that has been
introduced by David Gross in \cite{GrossMC} in a different context. We
prove the following lemma.
\begin{lemma}
\label{teo:kktdg}
Assume $\|\PO\PT\| \le 1/2$ and $\lambda < 1$.  Then with the same
notation, $(L_0, S_0)$ is the unique solution if there is a pair $(W,
F)$ obeying
\[
UV^* + W = \lambda(\sgn(S_0) + F + \PO D)
\]
with $\PT W = 0$ and $\|W\| \le \frac12$, $\PO F = 0$ and
$\|F\|_\infty \le \frac12$, and $\|\PO D\|_F \le \frac{1}{4}$.
\end{lemma}
\begin{proof} 
  Following the proof of Lemma \ref{teo:kkt}, we have
\begin{align*}
  \|L_0 + H\|_* + \lambda \|S_0 - H\|_1 & \ge \|L_0\|_* + \lambda
  \|S_0\|_1 + \frac12 \Bigl(\|\PTp H\|_* + \lambda \|\POc H\|_1\Bigr)
  - \lambda \<\PO D,
  H\>\\
  & \ge \|L_0\|_* + \lambda \|S_0\|_1 + \frac12 \Bigl(\|\PTp H\|_* +
  \lambda \|\POc H\|_1\Bigr) - {\lambda \over 4} \|\PO H\|_F.
\end{align*}
Observe now that 
\begin{align*}
  \|\PO H\|_F & \le \|\PO \PT H\|_F  + \|\PO \PTp H\|_F  \\
  & \le \frac12 \|H\|_F + \|\PTp H\|_F \\
  & \le \frac12\|\PO H\|_F + \frac12 \|\POc H\|_F  + \|\PTp H\|_F 
\end{align*}
and, therefore, 
\[
 \|\PO H\|_F \le  \|\POc H\|_F  + 2 \|\PTp H\|_F. 
\]
In conclusion, 
\[
\|L_0 + H\|_* + \lambda \|S_0 - H\|_1 \ge \|L_0\|_* + \lambda
\|S_0\|_1 + \frac12 \Bigl((1-\lambda) \|\PTp H\|_* + {\lambda \over 2}
\|\POc H\|_1\Bigr),
\]
and the term between parenthesis is strictly positive when $H \neq
0$. 
\end{proof}

As a consequence of Lemma \ref{teo:kktdg}, it now suffices to produce
a dual certificate $W$ obeying
\begin{equation}
\label{eq:dual-certif-dg}
\begin{cases} W \in T^\perp, \\
  \|W\| < 1/2,\\
  \|\PO(UV^* -\lambda \sgn(S_0) + W)\|_F \le \lambda/4,\\
  \|\POc(UV^* + W)\|_\infty < \lambda/2.\\
\end{cases}
\end{equation}
Further, we would like to note that the existing literature on matrix
completion \cite{CR08} gives good bounds on $\|\PO\PT\|$, see Theorem
\ref{teo:rudelson} in Section \ref{sec:keys}.

\subsection{Dual certification via the golfing scheme}
\label{sec:golfing}

In the papers \cite{GrossMC,GrossQuantum}, Gross introduces a new
scheme, termed the golfing scheme, to construct a dual certificate for
the matrix completion problem, i.e.~the problem of reconstructing a
low-rank matrix from a subset of its entries. In this section, we will
adapt this clever golfing scheme, with two important modifications, to
our separation problem.

Before we introduce our construction, our model assumes that $\Omega
\sim \text{Ber}(\rho)$, or equivalently that $\Omega^c \sim
\text{Ber}(1-\rho)$. Now the distribution of $\Omega^c$ is the same as
that of $\Omega^c = \Omega_1 \cup \Omega_2 \cup \ldots \cup
\Omega_{j_0}$, where each $\Omega_j$ follows the Bernoulli model with
parameter $q$, which has an explicit expression. To see this, observe
that by independence, we just need to make sure that any entry $(i,j)$
is selected with the right probability. We have
\[
\P((i,j) \in \Omega) = \P(\text{Bin}(j_0,q) = 0) = (1-q)^{j_0},
\]
so that the two models are the same if 
\[
\rho = (1-q)^{j_0},
\]
hence justifying our assertion. Note that because of overlaps between
the $\Omega_j$'s, $q \ge (1-\rho)/j_0$.




We now propose constructing a dual certificate $$W = W^L + W^S,$$ where
each component is as follows:
\begin{enumerate}
\item {\em Construction of $W^L$ via the golfing scheme.}  Fix an
  integer $j_0 \ge 1$ whose value shall be discussed later, and let
  $\Omega_j$, $1 \le j \le j_0$, be defined as above so that $\Omega^c
  = \cup_{1 \le j \le j_0} \Omega_j$.  Then starting with $Y_0 = 0$,
  inductively define
\[
Y_j = Y_{j-1} + q^{-1} \POj \PT(UV^* - Y_{j-1}),
\]
and set 
\begin{equation}
  \label{eq:WL}
  W^L = \PTp Y_{j_0}. 
\end{equation}
This is a variation on the golfing scheme discussed in \cite{GrossMC},
which assumes that the $\Omega_j$'s are sampled with replacement, and
does not use the projector $\POj$ but something more complicated
taking into account the number of times a specific entry has been
sampled.

\item {\em Construction of $W^S$ via the method of least squares.}
  Assume that $\|\PO\PT\| < 1/2$. Then $\|\PO \PT \PO\| < 1/4$ and,
  thus, the operator $\PO - \PO \PT \PO$ mapping $\Omega$ onto itself
  is invertible; we denote its inverse by $(\PO - \PO \PT \PO)^{-1}$.
  We then set
  \begin{equation}
    \label{eq:WS}
W^S = \lambda \PTp (\PO - \PO \PT \PO)^{-1} \sgn(S_0).
\end{equation}
Clearly, an equivalent definition is via the convergent Neumann series 
\begin{equation}
\label{eq:neuman}
W^S = \lambda \PTp \sum_{k \ge 0} (\PO \PT \PO)^k \sgn(S_0). 
\end{equation}
Note that $\PO W^S = \lambda \PO(I - \PT) (\PO - \PO \PT \PO)^{-1}
\sgn(S_0) = \lambda \sgn(S_0)$.  With this, the construction has a
natural interpretation: one can verify that among all matrices $W \in
T^\perp$ obeying $\PO W = \lambda \sgn(S_0)$, $W^S$ is that with
minimum Frobenius norm.
\end{enumerate}
Since both $W^L$ and $W^S$ belong to $T^\perp$ and $\PO W^S = \lambda
\sgn(S_0)$, we will establish that $W^L + W^S$ is a valid dual
certificate if it obeys
\begin{equation}
\label{eq:dual-certif-use}
\begin{cases}
  \|W^L + W^S\| < 1/2,\\
  \|\PO(UV^* + W^L)\|_F \le \lambda/4,\\
  \|\POc(UV^* + W^L + W^S)\|_\infty < \lambda/2.\\
\end{cases}
\end{equation}

\subsection{Key lemmas}
\label{sec:keys}

We now state three lemmas, which taken collectively, establish our
main theorem. The first may be found in \cite{CR08}.
\begin{theorem}\cite[Theorem 4.1]{CR08}
  \label{teo:rudelson} 
  Suppose $\Omega_0$ is sampled from the Bernoulli model with
  parameter $\rho_0$. Then with high probability,
  \begin{equation}
    \label{eq:rudelson} 
    \|\PT - \rho_0^{-1} \PT \POzero \PT \| \le \epsilon,  
  \end{equation}
  provided that $\rho_0 \ge C_0 \, \epsilon^{-2} \, \frac{\mu r \log
    n}{n}$ for some numerical constant $C_0 > 0$ ($\mu$ is the
  incoherence parameter). For rectangular matrices, we need $\rho_0
  \ge C_0 \, \epsilon^{-2} \, \frac{\mu r \log n_{(1)}}{n_{(2)}}$. 
\end{theorem}

Among other things, this lemma is important because it shows that
$\|\PO\PT\| \le 1/2$, provided $|\Omega|$ is not too large. Indeed, if
$\Omega \sim \text{Ber}(\rho)$, we have
\[
\|\PT - (1-\rho)^{-1} \PT \POc \PT\| \le \epsilon,  
\] 
with the proviso that $1-\rho \ge C_0 \, \epsilon^{-2} \, \frac{\mu r \log
  n}{n}$. Note, however, that since $\OpId = \PO + \POc$, 
\[
\PT - (1-\rho)^{-1} \PT \POc \PT = (1-\rho)^{-1}(\PT \PO \PT - \rho\PT)
\]
and, therefore, by the triangular inequality
\[
\|\PT \PO \PT\| \le \epsilon(1-\rho) + \rho \|\PT\| = \rho +
\epsilon(1-\rho). 
\]
Since $\|\PO\PT\|^2 = \|\PT\PO\PT\|$, we have established the
following: 
\begin{corollary}
  \label{teo:POPT}
  Assume that $\Omega \sim \text{Ber}(\rho)$, then $\|\PO\PT\|^2 \le
  \rho + \epsilon$, provided that $1-\rho \ge C_0 \, \epsilon^{-2} \,
  \frac{\mu r \log n}{n}$, where $C_0$ is as in Theorem
  \ref{teo:rudelson}.  For rectangular matrices, the modification is
  as in Theorem \ref{teo:rudelson}.
\end{corollary}

The lemma below is proved is Section \ref{sec:lemmas}. 
\begin{lemma}
  \label{teo:WL}
  Assume that $\Omega \sim \text{Ber}(\rho)$ with parameter $\rho \le
  \rho_s$ for some $\rho_s > 0$. Set $j_0 = 2\lceil \log n \rceil$ (use
  $\log n_{(1)}$ for rectangular matrices). Then under the other
  assumptions of Theorem \ref{teo:main}, the matrix $W^L$
  \eqref{eq:WL} obeys
\begin{enumerate}
\item[(a)] $\|W^L\| < 1/4$,
\item[(b)] $\|\PO(UV^* +  W^L)\|_F < \lambda/4$,
\item[(c)] $\|\POc(UV^* + W^L)\|_\infty < \lambda/4$.
\end{enumerate}
\end{lemma}

Since $\|\PO\PT\| < 1$ with large probability, $W^S$ is well defined
and the following holds.
\begin{lemma}
  \label{teo:WS}
  Assume that $S_0$ is supported on a set $\Omega$ sampled as in Lemma
  \ref{teo:WL}, and that the signs of $S_0$ are i.i.d.~symmetric (and
  independent of $\Omega$).  Then under the other assumptions of
  Theorem \ref{teo:main}, the matrix $W^S$ \eqref{eq:WS} obeys
\begin{enumerate}
\item[(a)] $\|W^S\| < 1/4$,
\item[(b)] $\|\POc W^S\|_\infty < \lambda/4$.
\end{enumerate}
\end{lemma}
The proof is also in Section \ref{sec:lemmas}.  Clearly, $W^L$ and
$W^S$ obey \eqref{eq:dual-certif-use}, hence certifying that \name\,
correctly recovers the low-rank and sparse components with high
probability when the signs of $S_0$ are random. The earlier
``derandomization'' argument then establishes Theorem \ref{teo:main}.

\section{Proofs of Dual Certification}
\label{sec:lemmas}

This section proves the two crucial estimates, namely, Lemma
\ref{teo:WL} and Lemma \ref{teo:WS}.

\subsection{Preliminaries}

We begin by recording two results which shall be useful in proving
Lemma \ref{teo:WL}.  While Theorem \ref{teo:rudelson} asserts that
with large probability,
\[
\|Z - \rho_0^{-1} \PT \POzero Z\|_F \le \epsilon \|Z\|_F,
\]
for all $Z \in T$, the next lemma shows that for a fixed $Z$, the
sup-norm of $Z - \rho_0^{-1} \PT \POzero(Z)$ also does not increase
(also with large probability).
\begin{lemma}
  \label{teo:infty} 
  Suppose $Z \in T$ is a fixed matrix, and $\Omega_0 \sim
  \text{Ber}(\rho_0)$. Then with  high probability,
   \begin{equation}
    \label{eq:infty} 
    \|Z - \rho_0^{-1} \PT \POzero Z\|_\infty \le \epsilon \|Z\|_\infty
  \end{equation}
  provided that $\rho_0 \ge C_0 \, \epsilon^{-2} \, \frac{\mu r \log
    n}{n}$ (for rectangular matrices, $\rho_0 \ge C_0 \, \epsilon^{-2}
  \, \frac{\mu r \log n_{(1)}}{n_{(2)}}$) for some numerical constant
  $C_0 > 0$.
\end{lemma}
The proof is an application of Bernstein's inequality and may be found
in the Appendix. A similar but somewhat different version of
\eqref{eq:infty} appears in \cite{RechtMC}.

The second result was proved in \cite{CR08}.
\begin{lemma}\cite[Theorem 6.3]{CR08}
\label{teo:sixthree}
  Suppose $Z$ is fixed, and $\Omega_0 \sim \text{Ber}(\rho_0)$. Then
  with high probability,
  \label{teo:CR} 
   \begin{equation}
    \label{eq:CR} 
    \|(I - \rho_0^{-1} \POzero) Z\| \le C'_0 \sqrt{\frac{n\log n}{\rho_0}} 
    \|Z\|_\infty
  \end{equation}
  for some small numerical constant $C'_0 > 0$ provided that $\rho_0
  \ge C_0 \, \frac{\mu \log n}{n}$ (or $\rho_0 \ge C'_0 \, \frac{\mu
    \log n_{(1)}}{n_{(2)}}$ for rectangular matrices in which case
  $n_{(1)} \log n_{(1)}$ replaces $n \log n$ in \eqref{eq:CR}).
\end{lemma}
As a remark, Lemmas \ref{teo:infty} and \ref{teo:sixthree}, and Theorem
\ref{teo:rudelson} all hold with probability at least $1 -
O(n^{-\beta})$, $\beta > 2$, if $C_0$ is replaced by $C \beta$ for
some numerical constant $C > 0$.

\subsection{Proof of Lemma \ref{teo:WL}}
\label{sec:WL}

We begin by introducing a piece of notation and set $Z_j = UV^* - \PT
Y_j$ obeying
\[
Z_j = (P_T - q^{-1} \PT \POj \PT) Z_{j-1}. 
\]
Obviously $Z_j \in T$ for all $j \ge 0$.  First, note that when
\begin{equation}
\label{eq:q}
q \ge C_0 \, \epsilon^{-2} \,
\frac{\mu r \log n}{n},
\end{equation}
(for rectangular matrices, take $q \ge C_0 \, \epsilon^{-2} \,
\frac{\mu r \log n_{(1)}}{n_{(2)}}$), we have
\begin{equation}
\|Z_j\|_\infty \le \epsilon \| Z_{j-1} \|_\infty  
\label{eq:infty-useful}
\end{equation}
by Lemma \ref{teo:infty}. (This holds with high probability
because $\Omega_j$ and $Z_{j-1}$ are independent, and this is why the
golfing scheme is easy to use.) In particular, this gives that with
high probability
\[
\|Z_j\|_\infty \le \epsilon^j \|UV^*\|_\infty. 
\]
When $q$ obeys the same estimate, 
\begin{equation}
\|Z_j\|_F \le \epsilon \|Z_{j-1}\|_F 
\label{eq:infty-useful2}
\end{equation}
by Theorem \ref{teo:rudelson}.  In particular, this gives that with
high probability
\begin{equation}
  \label{eq:F-useful}
  \|Z_j\|_F \le \epsilon^{j} \|UV^*\|_F = \epsilon^j \sqrt{r}. 
\end{equation}
Below, we will assume $\epsilon \le e^{-1}$.

\paragraph{Proof of (a).} We prove the first part of the lemma and the
argument parallels that in \cite{GrossMC}, see also
\cite{RechtMC}. From
\[
Y_{j_0} = \sum_j q^{-1} \POj Z_{j-1},
\]
we deduce
\begin{align*}
  \|W^L\| = \|\PTp Y_{j_0}\|_\infty &
  \le \sum_j  \|q^{-1} \PTp \POj Z_{j-1}\|\\
  & = \sum_j  \|\PTp (q^{-1} \POj Z_{j-1} - Z_{j-1})\|\\
  & \le \sum_j \|q^{-1} \POj Z_{j-1} - Z_{j-1}\|\\
  & \le C'_0 \sqrt{\frac{n\log n}{q}} \sum_j \|Z_{j-1}\|_\infty\\
  & \le C'_0 \sqrt{\frac{n\log n}{q}} \sum_j \epsilon^{j-1} \|UV^*\|_\infty\\
  & \le C'_0 (1-\epsilon)^{-1} \sqrt{\frac{n \log n}{q}}
  \|UV^*\|_\infty.
\end{align*}
The fourth step follows from Lemma \ref{teo:sixthree} and the fifth
from \eqref{eq:infty-useful2}. Since $\|UV^*\| \le \sqrt{\mu r}/{n}$,
this gives
\[
\|W^L\| \le C' \epsilon
\] 
for some numerical constant $C'$ whenever $q$ obeys
\eqref{eq:q}. 

\paragraph{Proof of (b).} Since $\PO Y_{j_0} = 0$, 
\[
\PO(UV^* + \PTp Y_{j_0}) = \PO(UV^* - \PT Y_{j_0}) = \PO(Z_{j_0}),   
\]
and it follows from \eqref{eq:F-useful} that
\[
\|Z_{j_0} \|_F \le \epsilon^{{j_0}} \|UV^*\|_F = \epsilon^{{j_0}} \sqrt{r}.  
\]
Since $\epsilon \le e^{-1}$ and $j_0 \ge 2\log n$, $\epsilon^{{j_0}}
\le 1/n^2$ and this proves the claim.

\paragraph{Proof of (c).} We have $UV^* + W^L = Z_{j_0} + Y_{j_0}$ and
know that $Y_{j_0}$ is supported on $\Omega^c$.  Therefore, since
$\|Z_{j_0}\|_F \le \lambda/8$, it suffices to show that
$\|Y_{j_0}\|_\infty \le \lambda/8$. We have
\begin{align*}
  \|Y_{j_0}\|_\infty & \le  q^{-1} \sum_j  \|\POj Z_{j-1}\|_\infty \\
  & \le q^{-1} \sum_j \|Z_{j-1}\|_\infty\\
  & \le q^{-1} \sum_j \epsilon^j \|UV^*\|_\infty.
\end{align*}
Since $\|UV^*\|_\infty \le \sqrt{\mu r}/{n}$, this gives
\[
\|Y_{j_0}\|_\infty \le C' \frac{\epsilon^2}{\sqrt{\mu r (\log n)^2}}
\]
for some numerical constant $C'$ whenever $q$ obeys \eqref{eq:q}.
Since $\lambda = 1/\sqrt{n}$, $\|Y_{j_0}\|_\infty \le \lambda/8$ if
\[
\epsilon \le C \Bigl( \frac{\mu r (\log n)^2}{n} \Bigr)^{1/4}.
\]

\paragraph{Summary.} We have seen that (a) and (b) are satisfied if
$\epsilon$ is sufficiently small and $j_0 \ge 2\log n$. For (c), we can
take $\epsilon$ on the order of $(\mu  r (\log n)^2/n)^{1/4}$,
which will be sufficiently small as well provided that $\rho_r$ in
\eqref{eq:main} is sufficiently small. Note that everything is
consistent since $C_0 \, \epsilon^{-2} \frac{\mu r \log n}{n} <
1$. This concludes the proof of Lemma \ref{teo:WL}.

\subsection{Proof of Lemma \ref{teo:WS}}

It is convenient to introduce the sign matrix $E = \sgn(S_0)$
distributed as
\begin{equation}
\label{eq:random-sign}
E_{ij} = \begin{cases} 1, & \text{w.~p. }
  \rho/2,\\
  0, & \text{w.~p. } 1-\rho,\\
  -1, & \text{w.~p. } \rho/2. 
\end{cases} 
\end{equation}
We shall be interested in the event $\{\|\PO\PT\| \le \sigma\}$ which
holds with large probability when $\sigma = \sqrt{\rho} + \epsilon$,
see Corollary \ref{teo:POPT}. In particular, for any $\sigma > 0$,
$\{\|\PO\PT\| \le \sigma\}$ holds with high probability provided
$\rho$ is sufficiently small. 

\paragraph{Proof of (a).} By construction, 
\begin{align*}
  W^S & = \lambda \PTp  E + \lambda \PTp \sum_{k \ge 1}
  (\PO\PT\PO)^k E\\
  & := \PTp W_0^S + \PTp W_1^S.
\end{align*}
For the first term, we have $\|\PTp W_0^S\| \le \|W_0^S\| = \lambda
\|E\|$. Then standard arguments about the norm of a matrix with
i.i.d.~entries give \cite{Vershynin280}
\[
\|E\| \le 4\sqrt{n \rho}
\]
with large probability.  Since $\lambda = 1/\sqrt{n}$, this gives
$\|W_0^S\| \le 4 \sqrt{\rho}$. When the matrix is rectangular, we have 
\[
\|E\| \le 4\sqrt{n_{(1)} \rho}
\]
with high probability. Since $\lambda = 1/\sqrt{n_{(1)}}$ in
this case, $\|W_0^S\| \le 4 \sqrt{\rho}$ as well.

Set $\mathcal{R} = \sum_{k \ge 1} (\PO\PT\PO)^{k}$ and observe that
$\mathcal{R}$ is self-adjoint. For the second term, $\|\PTp W_1^S\|
\le \|W_1^S\|$, where $W_1^S = \lambda \mathcal{R} (E)$.  We need to
bound the operator norm of the matrix $\mathcal{R}(E)$, and use a
standard covering argument to do this. Throughout, $N$ denotes an
$1/2$-net for $\mathbb{S}^{n-1}$ of size at most $6^n$ (such a net
exists, see \cite[Theorem 4.16]{Ledoux}). Then a standard argument
\cite{Vershynin280} shows that
\[
\|\mathcal{R}(E)\| = \sup_{x, y \in \mathbb{S}^{n-1}} \<y,
\mathcal{R}(E) x\> \le 4 \sup_{x,y \in N} \<y, \mathcal{R}(E) x\>.
\]
For a fixed pair $(x,y)$ of unit-normed vectors in $N \times N$,
define the random variable
\[
X(x,y) := \<y, \mathcal{R} (E) x\> = \<\mathcal{R}(yx^*), E\>. 
\]
Conditional on $\Omega = \text{supp}(E)$, the signs of $E$ are
i.i.d.~symmetric and Hoeffding's inequality gives
\[
\P(|X(x,y)| > t \, | \, \Omega) \le 2 \exp\Bigl(-\frac{2t^2}{
  \|\mathcal{R}(xy^*)\|_F^2}\Bigr).
\]
Now since $\|yx^*\|_F = 1$, the matrix $\mathcal{R}(yx^*)$ obeys $
\|\mathcal{R}(yx^*)\|_F \le \|\mathcal{R}\|$ and, therefore,
\[
\P\Bigl(\sup_{x,y \in N} |X(x,y)| > t \, | \, \Omega\Bigr) \le 2|N|^2 
\exp\Bigl(-\frac{2t^2}{\|\mathcal{R}\|^2}\Bigr).
\]
Hence, 
\[
\P(\|\mathcal{R}(E)\| > t \, | \, \Omega) \le 2|N|^2
\exp\Bigl(-\frac{t^2}{8 \|\mathcal{R}\|^2}\Bigr). 
\]
On the event $\{\|\PO\PT\| \le \sigma\}$,
\[
\|\mathcal{R}\| \le \sum_{k \ge 1} \sigma^{2k} =
\frac{\sigma^2}{1-\sigma^2}
\]
and, therefore, unconditionally,
\[
\P(\|\mathcal{R}(E)\| > t) \le 2 |N|^2 \, \exp\Bigl(- \frac{\gamma^2
  t^2}{2}\Bigr) + \P(\|\PO\PT\| \ge \sigma), \qquad \gamma =
\frac{1-\sigma^2}{2\sigma^2}.
\]
This gives 
\[
\P(\lambda \|\mathcal{R}(E)\| > t) \le 2 \times 6^{2n} \exp\Bigl(-
\frac{\gamma^2 t^2}{2\lambda^2}\Bigr) + \P(\|\PO\PT\| \ge \sigma).
\]
With $\lambda = 1/\sqrt{n}$,
\[
\|W^S\| \le 1/4,
\]
with large probability, provided that $\sigma$, or equivalently
$\rho$, is small enough.

\paragraph{Proof of (b).} Observe that 
\[
\POc W^S  = -\lambda \POc \PT (\PO - \PO \PT \PO)^{-1} E.
\]
Now for $(i,j) \in \Omega^c$, $W^S_{ij} = \<e_i, W^S e_j\> = \<e_i
e_j^*, W^S\>$, and we have 
\[
W^S_{ij} = \lambda \<X(i,j), E\>, 
\]
where $X(i,j)$ is the matrix $-(\PO - \PO \PT \PO)^{-1} \PO
\PT(e_ie_j^*)$.  Conditional on $\Omega = \text{supp}(E)$, the signs
of $E$ are i.i.d.~symmetric, and Hoeffding's inequality gives
\[
\P( |W^S_{ij}| > t \lambda \, | \, \Omega) \le 2
\exp\Bigl(-\frac{2t^2}{ \|X(i,j)\|_F^2}\Bigr), 
\]
and, thus,
\[
\P\Bigl(\sup_{i,j} |W^S_{ij}| > t \lambda \, | \, \Omega\Bigr) \le 2n^2
\exp\Bigl(-\frac{2t^2}{\sup_{i,j} \|X(i,j)\|_F^2}\Bigr).
\] 
Since \eqref{eq:PTeiej} holds, we have
\[
\|\PO \PT(e_ie_j^*)\|_F \le \|\PO \PT\| \|\PT(e_ie_j^*)\|_F \le \sigma
\sqrt{2\mu r/n}
\]
on the event $\{\|\PO\PT\| \le \sigma\}$.  On the same event, $\|(\PO -
\PO \PT \PO)^{-1}\| \le (1-\sigma^2)^{-1}$ and, therefore, 
\[
\|X(i,j)\|^2_F \le \frac{2\sigma^2}{(1-\sigma^2)^2} \, \frac{\mu r}{n}.
\]
Then unconditionally,
\[
\P\Bigl(\sup_{i,j} |W^S_{ij}| > t \lambda\Bigr) \le 2n^2
\exp\Bigl(-\frac{n \gamma^2 t^2}{\mu r}\Bigr)+ \P(\|\PO\PT\| \ge
\sigma), \qquad \gamma = \frac{(1-\sigma^2)^2}{2\sigma^2}.
\]
This proves the claim when $\mu r < \rho'_r n (\log n)^{-1}$ and
$\rho'_r$ is sufficiently small.

\section{Numerical Experiments and
  Applications}\label{sec:experiments}

In this section, we perform numerical experiments corroborating our
main results and suggesting their many applications in image and video
analysis. We first investigate \name's ability to
correctly recover matrices of various rank from errors of various
density. We then sketch applications in background modeling from video
and removing shadows and specularities from face images. 

While the exact recovery guarantee provided by Theorem \ref{teo:main}
is independent of the particular algorithm used to solve \name, its
applicability to large scale problems depends on the availability of
scalable algorithms for nonsmooth convex optimization. For the
experiments in this section, we use the an augmented Lagrange
multiplier algorithm introduced in
\cite{Lin2009-ALM,Yuan2009}.\footnote{Both \cite{Lin2009-ALM,Yuan2009}
  have posted a version of their code online.} In Section
\ref{sec:algorithms}, we describe this algorithm in more detail, and
explain why it is our algorithm of choice for sparse and low-rank
separation.

One important implementation detail in our approach is the choice of
$\lambda$. Our analysis identifies one choice, $\lambda =
1/\sqrt{\text{max}(n_1,n_2)}$, which works well for incoherent
matrices. In order to illustrate the theory, throughout this section
we will always choose $\lambda = 1/\sqrt{\text{max}(n_1,n_2)}$. For
practical problems, however, it is often possible to improve
performance by choosing $\lambda$ according to prior knowledge about
the solution. For example, if we know that $S$ is very sparse,
increasing $\lambda$ will allow us to recover matrices $L$ of larger
rank. For practical problems, we recommend $\lambda =
1/\sqrt{\text{max}(n_1,n_2)}$ as a good rule of thumb, which can then be
adjusted slightly to obtain the best possible result.

\subsection{Exact recovery from varying fractions of error}

We first verify the correct recovery phenomenon of Theorem \ref{teo:main} on
randomly generated problems. We consider square matrices of varying
dimension $n = 500, \ldots, 3000$. We generate a rank-$r$ matrix
$L_0$ as a product $L_0 = X Y^*$ where $X$ and $Y$ are $n \times r$
matrices with entries independently sampled from a
$\mathcal{N}(0,1/n)$ distribution. $S_0$ is generated by choosing a
support set $\Omega$ of size $k$ uniformly at random, and setting $S_0
= \PO E$, where $E$ is a matrix with independent Bernoulli $\pm 1$
entries.

Table \ref{tab:exact-recovery} (top) reports the results with $r =
\mathrm{rank}(L_0) = 0.05 \times n$ and $k = \| S_0 \|_0 = 0.05 \times
n^2$.  Table \ref{tab:exact-recovery} (bottom) reports the results for
a more challenging scenario, $\mathrm{rank}(L_0) = 0.05 \times n$ and
$k = 0.10 \times n^2$. In all cases, we set $\lambda =
1/\sqrt{n}$. Notice that in all cases, solving the convex PCP
gives a result $(L, S)$ with the correct rank and sparsity. Moreover,
the relative error $\| L - L_0 \|_F / \| L_0\|_F$ is small, less than
$10^{-5}$ in all examples considered.\footnote{We measure relative
  error in terms of $L$ only, since in this paper we view the sparse
  and low-rank decomposition as recovering a low-rank matrix $L_0$
  from gross errors. $S_0$ is of course also well-recovered: in
  this example, the relative error in $S$ is actually smaller than
  that in $L$.}

The last two columns of Table \ref{tab:exact-recovery} give the number
of partial singular value decompositions computed in the course of the
optimization ($\#$ SVD) as well as the total computation time. This
experiment was performed in Matlab on a Mac Pro with dual quad-core
2.66 GHz Intel Xenon processors and 16 GB RAM. As we will discuss in
Section \ref{sec:algorithms} the dominant cost in solving the convex
program comes from computing one partial SVD per
iteration. Strikingly, in Table \ref{tab:exact-recovery}, the number
of SVD computations is nearly constant regardless of dimension, 
and in all cases less than 17.\footnote{One might reasonably ask whether this near constant
  number of iterations is due to the fact that random problems are in
  some sense well-conditioned. There is some validity to this concern,
  as we will see in our real data examples. \cite{Lin2009-ALM}
  suggests a continuation strategy (there termed ``Inexact ALM'') that
  produces qualitatively similar solutions with a similarly small
  number of iterations. However, to the best of our knowledge its
  convergence is not guaranteed.} This suggests that in addition to
being theoretically well-founded, the recovery procedure advocated in
this paper is also reasonably practical.

\begin{table}
\begin{centering}
\begin{tabular}{|c|c|c|c|c|c|c|c|}
\hline
Dimension $n$ & $\mathrm{rank}(L_0)$ & $\| S_0 \|_0$ & $\mathrm{rank}(\hat L)$ & $\| \hat S \|_0$ & $\frac{\| \hat L - L_0 \|_F}{ \|L_0\|_F }$ & $\#$ SVD & Time(s) \\
\hline
\hline
500 & 25 & 12,500 & 25 & 12,500 & $1.1\times 10^{-6}$ & 16 & 2.9 \\
\hline
1,000 & 50 & 50,000 & 50 & 50,000 & $1.2\times 10^{-6}$ & 16 & 12.4 \\
\hline
2,000 & 100 & 200,000 & 100 & 200,000 & $1.2 \times 10^{-6}$& 16 & 61.8\\
\hline
3,000 & 250 & 450,000 & 250 & 450,000 & $2.3 \times 10^{-6}$ & 15 & 185.2 \\
\hline
\end{tabular} \\
$\rank(L_0) = 0.05 \times n$, $\| S_0 \|_0 = 0.05 \times n^2$. \\
\vspace{.1in}
\begin{tabular}{|c|c|c|c|c|c|c|c|}
\hline
Dimension $n$ & $\mathrm{rank}(L_0)$ & $\| S_0 \|_0$ & $\mathrm{rank}(\hat L)$ & $\| \hat S \|_0$ & $\frac{\| \hat L - L_0 \|_F}{ \|L_0\|_F }$ & $\#$ SVD & Time(s) \\
\hline
\hline
500 & 25 & 25,000 & 25 & 25,000 & $1.2\times 10^{-6}$ & 17 & 4.0 \\
\hline
1,000 & 50 & 100,000 & 50 & 100,000 & $2.4\times 10^{-6}$& 16 & 13.7 \\
\hline
2,000 & 100 & 400,000 & 100 & 400,000 & $2.4 \times 10^{-6}$ & 16 & 64.5 \\
\hline
3,000 & 150 & 900,000 & 150 & 900,000 & $2.5 \times 10^{-6}$ & 16 & 191.0 \\
\hline
\end{tabular} \\
$\rank(L_0) = 0.05 \times n$, $\| S_0 \|_0 = 0.10 \times n^2$. \\ \vspace{.1in}
\end{centering}
\caption{Correct recovery for random problems of varying size. 
  Here, $L_0 = XY^* \in \R^{n \times n}$ with $X,Y \in \mathbb{R}^{n \times r}$; $X,Y$ have entries i.i.d.~$\mathcal{N}(0,1/n)$. $S_0 \in \{-1,0,1\}^{n \times n}$ has support chosen uniformly at random and independent random signs; $\|S_0\|_0$ is the number of nonzero entries in $S_0$.   
  Top: recovering matrices of rank $0.05 \times n$ from $5\%$ gross errors. Bottom: recovering matrices of rank $0.05 \times n$ from $10\%$ gross errors. In all cases, the rank of $L_0$ and $\ell_0$-norm of $S_0$ are correctly estimated. Moreover, the number of partial singular value decompositions ($\#$ SVD) required to solve PCP is almost constant.} \label{tab:exact-recovery}
\end{table}

\subsection{Phase transition in rank and sparsity}
Theorem \ref{teo:main} shows that convex programming correctly recovers an
incoherent low-rank matrix from a constant fraction $\rho_s$ of
errors. We next empirically investigate the algorithm's ability to
recover matrices of varying rank from errors of varying sparsity. We
consider square matrices of dimension $n_1 = n_2 = 400$. We generate
low-rank matrices $L_0 = X Y^*$ with $X$ and $Y$ independently chosen
$n \times r$ matrices with i.i.d.~Gaussian entries of mean zero and
variance $1/n$. For our first experiment, we assume a Bernoulli model
for the support of the sparse term $S_0$, with random signs: each
entry of $S_0$ takes on value $0$ with probability $1-\rho$, and
values $\pm 1$ each with probability $\rho/2$. For each $(r,\rho)$
pair, we generate $10$ random problems, each of which is solved via
the algorithm of Section \ref{sec:algorithms}. We declare a trial to
be successful if the recovered $\hat L$ satisfies $\| L - L_0 \|_F /
\| L_0 \|_F \le 10^{-3}$. Figure \ref{fig:phase-transition} (left)
plots the fraction of correct recoveries for each pair
$(r,\rho)$. Notice that there is a large region in which the recovery
is exact. This highlights an interesting aspect of our result: the
recovery is correct even though in some cases $\|S_0 \|_F \gg \| L_0
\|_F$ (e.g., for $r/n = \rho$, $\|S_0\|_F$ is $\sqrt{n} = 20$ times
larger!). This is to be expected from Lemma \ref{teo:kkt}: the existence (or
non-existence) of a dual certificate depends only on the signs and
support of $S_0$ and the orientation of the singular spaces of $L_0$.

However, for incoherent $L_0$, our main result goes one step further
and asserts that the signs of $S_0$ are also not important: recovery
can be guaranteed as long as its support is chosen uniformly at
random. We verify this by again sampling $L_0$ as a product of
Gaussian matrices and choosing the support $\Omega$ according to the
Bernoulli model, but this time setting $S_0 = \PO
\mathrm{sgn}(L_0)$. One might expect such $S_0$ to be more difficult
to distinguish from $L_0$. Nevertheless, our analysis showed that the
number of errors that can be corrected drops by at most $1/2$ when
moving to this more difficult model. Figure \ref{fig:phase-transition}
(middle) plots the fraction of correct recoveries over $10$ trials,
again varying $r$ and $\rho$. Interestingly, the region of correct
recovery in Figure \ref{fig:phase-transition} (middle) actually
appears to be broader than that in Figure \ref{fig:phase-transition}
(left). Admittedly, the shape of the region in the upper-left corner
is puzzling, but has been `confirmed' by several distinct simulation
experiments (using different solvers). 

Finally, inspired by the connection between matrix completion and
robust PCA, we compare the breakdown point for the low-rank and sparse
separation problem to the breakdown behavior of the nuclear-norm
heuristic for matrix completion. By comparing the two heuristics, we
can begin to answer the question {\em how much is gained by knowing
  the location $\Omega$ of the corrupted entries}? Here, we again
generate $L_0$ as a product of Gaussian matrices. However, we now
provide the algorithm with only an incomplete subset $M = \POc L_0$ of
its entries. Each $(i,j)$ is included in $\Omega$ independently with
probability $1-\rho$, so rather than a probability of error, here,
$\rho$ stands for the probability that an entry is omitted. We solve
the nuclear norm minimization problem
\[
    \text{minimize} \quad \|L\|_* \quad \text{subject to} \quad \POc L = \POc M 
\]
using an augmented Lagrange multiplier algorithm very similar to the
one discussed in Section \ref{sec:algorithms}. We again declare $L_0$
to be successfully recovered if $\| L - L_0 \|_F / \| L_0 \|_F <
10^{-3}$. Figure \ref{fig:phase-transition} (right) plots the fraction
of correct recoveries for varying $r,\rho$. Notice that nuclear norm
minimization successfully recovers $L_0$ over a much wider range of
$(r,\rho)$. This is interesting because in the regime of large $k$, $k
= \Omega(n^2)$, the best performance guarantees for each heuristic
agree in their order of growth -- both guarantee correct recovery for
$\mathrm{rank}(L_0) = O(n/\log^2 n)$. Fully explaining the difference
in performance between the two problems may require a sharper analysis
of the breakdown behavior of each.

\begin{figure}
\begin{center}
\subfigure[Robust PCA, Random Signs]{\includegraphics[width=2in]{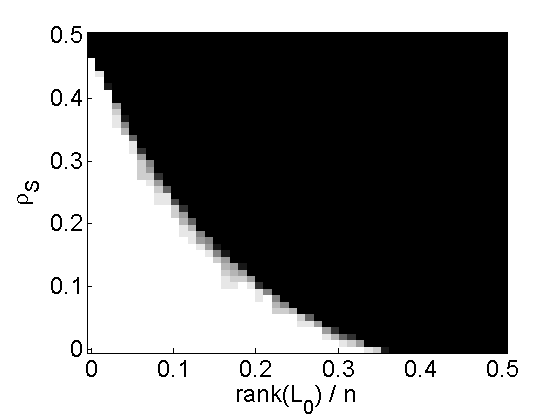}} 
\subfigure[Robust PCA, Coherent Signs]{\includegraphics[width=2in]{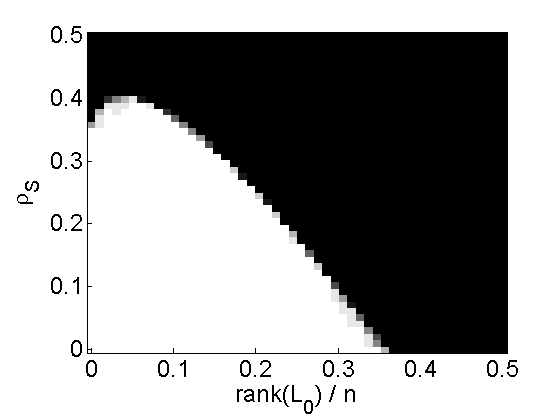}}
\subfigure[Matrix Completion]{\includegraphics[width=2in]{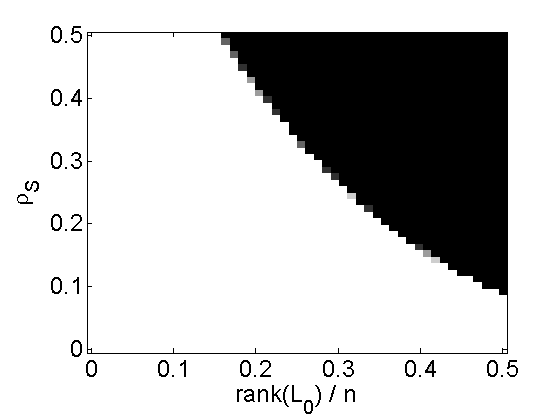}}
\end{center}
\caption{{\bf Correct recovery for varying rank and sparsity.} Fraction of correct recoveries across 10 trials, as a function of $\mathrm{rank}(L_0)$ (x-axis) and sparsity of $S_0$ (y-axis). Here, $n_1 = n_2 = 400$. In all cases, $L_0 = XY^*$ is a product of independent $n \times r$ i.i.d.~$\mathcal{N}(0,1/n)$ matrices. Trials are considered successful if $\| \hat L - L_0 \|_F / \| L_0 \|_F < 10^{-3}$. Left: low-rank and sparse decomposition, $\sgn(S_0)$ random. Middle: low-rank and sparse decomposition, $S_0 = \PO \sgn(L_0)$. Right: matrix completion. For matrix completion, $\rho_s$ is the probability that an entry is omitted from the observation.} \label{fig:phase-transition}
\end{figure}

\subsection{Application sketch: background modeling from surveillance
  video} Video is a natural candidate for low-rank modeling, due to
the correlation between frames. One of the most basic algorithmic
tasks in video surveillance is to estimate a good model for the
background variations in a scene. This task is complicated by the
presence of foreground objects: in busy scenes, every frame may
contain some anomaly. Moreover, the background model needs to be
flexible enough to accommodate changes in the scene, for example due to
varying illumination. In such situations, it is natural to model the
background variations as approximately low rank. Foreground objects,
such as cars or pedestrians, generally occupy only a fraction of the
image pixels and hence can be treated as sparse errors.

We investigate whether convex optimization can separate these sparse
errors from the low-rank background. Here, it is important to note
that the error support may not be well-modeled as Bernoulli: errors
tend to be spatially coherent, and more complicated models such as
Markov random fields may be more appropriate
\cite{Cevher2008-ECCV,Zhou2009-ICCV}. Hence, our theorems do not
necessarily guarantee the algorithm will succeed with high
probability. Nevertheless, as we will see, \name\, still
gives visually appealing solutions to this practical low-rank and
sparse separation problem, without using any additional information
about the spatial structure of the error.

We consider two example videos introduced in \cite{Li2004-TIP}. The
first is a sequence of $200$ grayscale frames taken in an
airport. This video has a relatively static background, but
significant foreground variations. The frames have resolution $176
\times 144$; we stack each frame as a column of our matrix $M \in
\mathbb{R}^{25,344 \times 200}$. We decompose $M$ into a low-rank term
and a sparse term by solving the convex PCP problem \eqref{eq:sdp}
with $\lambda = 1/\sqrt{n_{1}}$. On a desktop PC with a 2.33 GHz Core2
Duo processor and 2 GB RAM, our Matlab implementation requires 806
iterations, and roughly 43 minutes to converge.\footnote{The paper
  \cite{Lin2009-ALM} suggests a variant of ALM optimization procedure,
  there termed the ``Inexact ALM'' that finds a visually similar
  decomposition in far fewer iterations (less than 50). However, since
  the convergence guarantee for that variant is weak, we choose to
  present the slower, exact result here.} Figure
\ref{fig:surveillance-1}(a) shows three frames from the video; (b) and
(c) show the corresponding columns of the low rank matrix $\hat L$ and
sparse matrix $\hat S$ (its absolute value is shown here). Notice that
$\hat L$ correctly recovers the background, while $\hat S$ correctly
identifies the moving pedestrians. The person appearing in the images
in $\hat L$ does not move throughout the video.

\begin{figure}[t]
\begin{center}
\subfigure{\includegraphics[scale=.46]{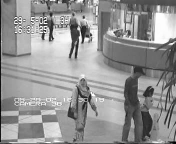}} \addtocounter{subfigure}{-1}
\subfigure{\includegraphics[scale=.46]{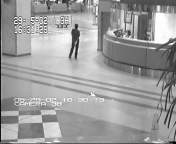}} \addtocounter{subfigure}{-1}
\subfigure{\includegraphics[scale=.46]{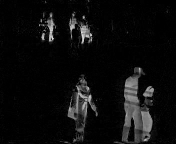}} \addtocounter{subfigure}{-1} 
\subfigure{\includegraphics[scale=.46]{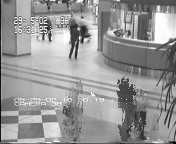}} \addtocounter{subfigure}{-1}
\subfigure{\includegraphics[scale=.46]{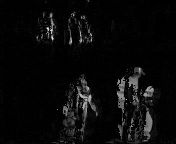}} \addtocounter{subfigure}{-1} \\
\subfigure{\includegraphics[scale=.46]{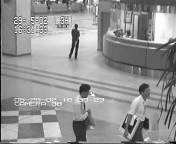}} \addtocounter{subfigure}{-1}
\subfigure{\includegraphics[scale=.46]{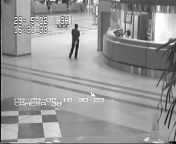}} \addtocounter{subfigure}{-1}
\subfigure{\includegraphics[scale=.46]{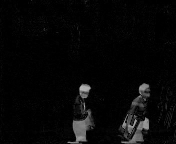}} \addtocounter{subfigure}{-1}
\subfigure{\includegraphics[scale=.46]{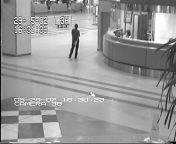}} \addtocounter{subfigure}{-1}
\subfigure{\includegraphics[scale=.46]{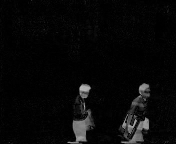}} \addtocounter{subfigure}{-1} \\
\subfigure[Original frames]{\includegraphics[scale=.46]{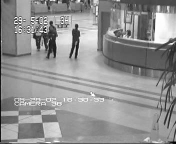}} \addtocounter{subfigure}{0}
\subfigure[Low-rank $\hat L$]{\includegraphics[scale=.46]{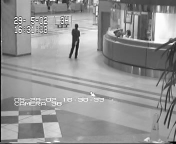}} \addtocounter{subfigure}{0}
\subfigure[Sparse $\hat S$]{\includegraphics[scale=.46]{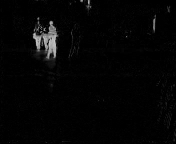}} \addtocounter{subfigure}{0}
\subfigure[Low-rank $\hat L$]{\includegraphics[scale=.46]{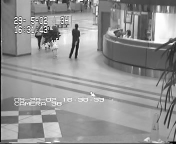}} \addtocounter{subfigure}{0}
\subfigure[Sparse $\hat S$]{\includegraphics[scale=.46]{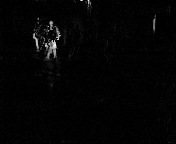}} \addtocounter{subfigure}{0} \\
\centerline{\small \hspace{1in}{Convex optimization (this work)}
\hspace{.8in}{Alternating minimization \cite{DeLaTorre2003-IJCV}}}
\end{center}
\caption{Background modeling from video. Three frames from a 200 frame
  video sequence taken in an airport \cite{Li2004-TIP}. (a) Frames of
  original video $M$. (b)-(c) Low-rank $\hat L$ and sparse components
  $\hat S$ obtained by PCP, (d)-(e) competing approach
  based on alternating minimization of an $m$-estimator
  \cite{DeLaTorre2003-IJCV}. PCP yields a much more appealing result
  despite using less prior knowledge.} \label{fig:surveillance-1}
\end{figure}

\begin{figure}[t]
\begin{center}
\subfigure{\includegraphics[scale=.50]{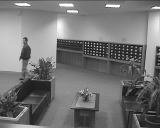}} \addtocounter{subfigure}{-1}
\subfigure{\includegraphics[scale=.50]{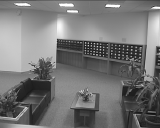}} \addtocounter{subfigure}{-1}
\subfigure{\includegraphics[scale=.50]{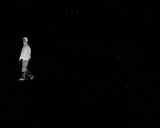}} \addtocounter{subfigure}{-1} 
\subfigure{\includegraphics[scale=.50]{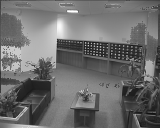}} \addtocounter{subfigure}{-1}
\subfigure{\includegraphics[scale=.50]{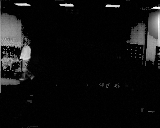}} \addtocounter{subfigure}{-1} \\
\subfigure{\includegraphics[scale=.50]{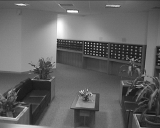}} \addtocounter{subfigure}{-1}
\subfigure{\includegraphics[scale=.50]{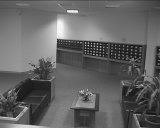}} \addtocounter{subfigure}{-1}
\subfigure{\includegraphics[scale=.50]{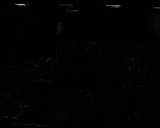}} \addtocounter{subfigure}{-1}
\subfigure{\includegraphics[scale=.50]{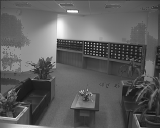}} \addtocounter{subfigure}{-1}
\subfigure{\includegraphics[scale=.50]{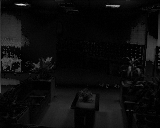}} \addtocounter{subfigure}{-1} \\
\subfigure[Original frames]{\includegraphics[scale=.50]{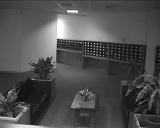}} \addtocounter{subfigure}{0}
\subfigure[Low-rank $\hat L$]{\includegraphics[scale=.50]{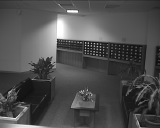}} \addtocounter{subfigure}{0}
\subfigure[Sparse $\hat S$]{\includegraphics[scale=.50]{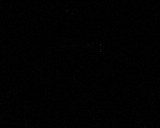}} \addtocounter{subfigure}{0}
\subfigure[Low-rank $\hat L$]{\includegraphics[scale=.50]{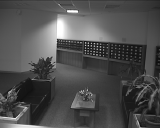}} \addtocounter{subfigure}{0}
\subfigure[Sparse $\hat S$]{\includegraphics[scale=.50]{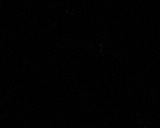}} \addtocounter{subfigure}{0} \\
\centerline{\small \hspace{1in}{Convex optimization (this work)}
\hspace{.8in}{Alternating minimization \cite{DeLaTorre2003-IJCV}}}
\end{center}
\caption{Background modeling from video. Three frames from a 250 frame
  sequence taken in a lobby, with varying illumination
  \cite{Li2004-TIP}. (a) Original video $M$. (b)-(c) Low-rank $\hat L$
  and sparse $\hat S$ obtained by PCP. (d)-(e) Low-rank
  and sparse components obtained by a competing approach based on
  alternating minimization of an m-estimator
  \cite{DeLaTorre2003-IJCV}. Again, convex programming yields a more
  appealing result despite using less prior
  information.} \label{fig:surveillance-2}
\end{figure}

Figure \ref{fig:surveillance-1} (d) and (e) compares the result
obtained by \name\, to a state-of-the-art technique from
the computer vision literature, \cite{DeLaTorre2003-IJCV}.\footnote{We
  use the code package downloaded from
  \url{http://www.salleurl.edu/~ftorre/papers/rpca/rpca.zip}, modified
  to choose the rank of the approximation as suggested in
  \cite{DeLaTorre2003-IJCV}.} That approach also aims at robustly
recovering a good low-rank approximation, but uses a more complicated,
nonconvex $m$-estimator, which incorporates a local scale estimate
that implicitly exploits the spatial characteristics of natural
images. This leads to a highly nonconvex optimization, which is
solved locally via alternating minimization. Interestingly, despite
using more prior information about the signal to be recovered, this
approach does not perform as well as the convex programming heuristic:
notice the large artifacts in the top and bottom rows of Figure
\ref{fig:surveillance-1} (d).

In Figure \ref{fig:surveillance-2}, we consider $250$ frames of a
sequence with several drastic illumination changes. Here, the
resolution is $168 \times 120$, and so $M$ is a $20,160 \times 250$
matrix. For simplicity, and to illustrate the theoretical results
obtained above, we again choose $\lambda = 1/\sqrt{n_1}$.\footnote{For
  this example, slightly more appealing results can actually be
  obtained by choosing larger $\lambda$ (say, $2 / \sqrt{n_1}$).} For
this example, on the same 2.66 GHz Core 2 Duo machine, the algorithm
requires a total of 561 iterations and 36 minutes to converge.

Figure \ref{fig:surveillance-2} (a) shows three frames taken from the
original video, while (b) and (c) show the recovered low-rank and
sparse components, respectively. Notice that the low-rank component
correctly identifies the main illuminations as background, while the
sparse part corresponds to the motion in the scene. On the other hand,
the result produced by the algorithm of \cite{DeLaTorre2003-IJCV}
treats some of the first illumination as foreground. PCP again
outperforms the competing approach, despite using less prior
information. These results suggest the potential power for convex
programming as a tool for video analysis.

Notice that the number of iterations for the real data is typically
higher than that of the simulations with random matrices given in
Table \ref{tab:exact-recovery}. The reason for this discrepancy might
be that the structures of real data could slightly deviate from the
idealistic low-rank and sparse model. Nevertheless, it is important to
realize that practical applications such as video surveillance often
provide additional information about the signals of interest, e.g. the
support of the sparse foreground is spatially piecewise contiguous, or
even impose additional requirements, e.g. the recovered background
needs to be non-negative etc. We note that the simplicity of our
objective and solution suggests that one can easily incorporate
additional constraints and more accurate models of the signals so as
to obtain much more efficient and accurate solutions in the future.


\subsection{Application sketch: removing shadows and specularities from face images} 
Face recognition is another problem domain in computer vision where
low-dimensional linear models have received a great deal of
attention. This is mostly due to the work of Basri and Jacobs, who
showed that for convex, Lambertian objects, images taken under distant
illumination lie near an approximately nine-dimensional linear
subspace known as the {\em harmonic plane}
\cite{Basri2003-PAMI}. However, since faces are neither perfectly
convex nor Lambertian, real face images often violate this low-rank
model, due to cast shadows and specularities. These errors are large
in magnitude, but sparse in the spatial domain. It is reasonable to
believe that if we have enough images of the same face, \name\, will
be able to remove these errors. As with the previous example, some
caveats apply: the theoretical result suggests the performance should
be good, but does not guarantee it, since again the error support does
not follow a Bernoulli model. Nevertheless, as we will see, the
results are visually striking.
\begin{figure}[t]
\centerline{
\begin{minipage}{3in}
\centering
\subfigure{
\includegraphics[scale=0.33]{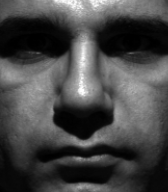}
}
\addtocounter{subfigure}{-1}
\subfigure{
\includegraphics[scale=0.33]{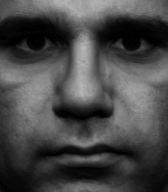}
}
\addtocounter{subfigure}{-1}
\subfigure{
\includegraphics[scale=0.33]{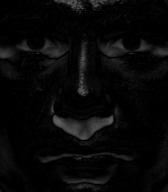}
} \addtocounter{subfigure}{-1} \\
\subfigure{
\includegraphics[scale=0.33]{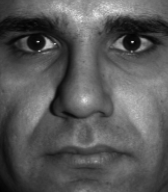}
} \addtocounter{subfigure}{-1}
\subfigure{
\includegraphics[scale=0.33]{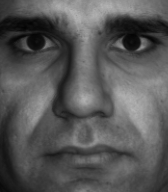}
} \addtocounter{subfigure}{-1}
\subfigure{
\includegraphics[scale=0.33]{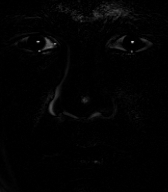}
} \addtocounter{subfigure}{-1} \\
\centering
\subfigure[$M$]{
\includegraphics[scale=0.33]{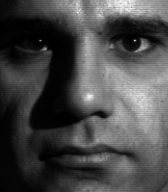}
} \addtocounter{subfigure}{0}
\subfigure[$\hat L$]{
\includegraphics[scale=0.33]{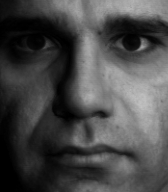}
} \addtocounter{subfigure}{0}
\subfigure[$\hat S$]{
\includegraphics[scale=0.33]{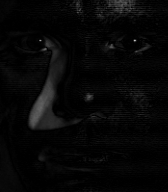}
}
\end{minipage}
\vline
\begin{minipage}{3in}
\centering
\addtocounter{subfigure}{-3}
\subfigure{
\includegraphics[scale=0.33]{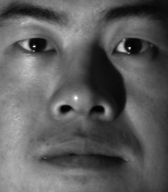}
}
\addtocounter{subfigure}{-1}
\subfigure{
\includegraphics[scale=0.33]{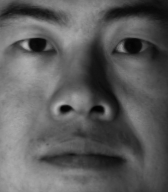}
}
\addtocounter{subfigure}{-1}
\subfigure{
\includegraphics[scale=0.33]{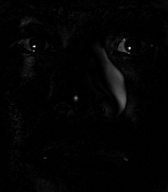}
} \addtocounter{subfigure}{-1} \\
\subfigure{
\includegraphics[scale=0.33]{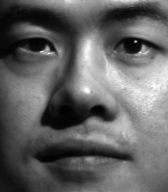}
} \addtocounter{subfigure}{-1}
\subfigure{
\includegraphics[scale=0.33]{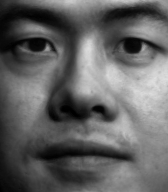}
} \addtocounter{subfigure}{-1}
\subfigure{
\includegraphics[scale=0.33]{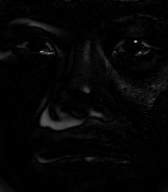}
} \addtocounter{subfigure}{-1} \\
\centering
\subfigure[$M$]{
\includegraphics[scale=0.33]{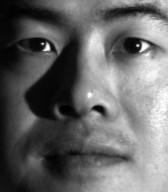}
} \addtocounter{subfigure}{0}
\subfigure[$\hat L$]{
\includegraphics[scale=0.33]{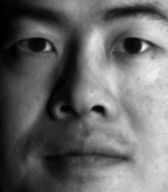}
} \addtocounter{subfigure}{0}
\subfigure[$\hat S$]{
  \includegraphics[scale=0.33]{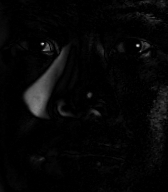}
}
\end{minipage}
}
\caption{Removing shadows, specularities, and saturations from face images. (a)
  Cropped and aligned images of a person's face under different
  illuminations from the Extended Yale B database. The size of each
  image is $192 \times 168$ pixels, a total of 58 different
  illuminations were used for each person. (b) Low-rank approximation
  $\hat L$ recovered by convex programming. (c) Sparse error $\hat S$
  corresponding to specularities in the eyes, shadows around the nose
  region, or brightness saturations on the face. Notice in the bottom
  left that the sparse term also compensates for errors in image
  acquisition.} \label{fig:faces}
\end{figure}

Figure \ref{fig:faces} shows two examples with face images taken
from the Yale B face database \cite{Georghiades2001-PAMI}. Here, each
image has resolution $192 \times 168$; there are a total of $58$
illuminations per subject, which we stack as the columns of our matrix
$M \in \mathbb{R}^{32,256 \times 58}$. We again solve PCP 
with $\lambda = 1/\sqrt{n_1}$. In this case, the algorithm requires
642 iterations to converge, and the total computation time on the same
Core 2 Duo machine is $685$ seconds.

Figure \ref{fig:faces} plots the low rank term $\hat L$ and the
magnitude of the sparse term $\hat S$ obtained as the solution to the
convex program. The sparse term $\hat S$ compensates for cast shadows
and specular regions. In one example (bottom row of Figure
\ref{fig:faces} left), this term also compensates for errors in image
acquisition. These results may be useful for conditioning the training
data for face recognition, as well as face alignment and tracking
under illumination variations.


\section{Algorithms} \label{sec:algorithms}

Theorem \ref{teo:main} shows that incoherent low-rank matrices can be
recovered from nonvanishing fractions of gross errors in polynomial
time. Moreover, as the experiments in the previous section attest, the
low computation cost is guaranteed not only in theory, the efficiency
is becoming {\em practical} for real imaging problems. This
practicality is mainly due to the rapid recent progress in scalable
algorithms for nonsmooth convex optimization, in particular for
minimizing the $\ell_1$ and nuclear norms. In this section, we briefly
review this progress, and discuss our algorithm of choice for this
problem.

For small problem sizes, \name
\[
  \begin{array}{ll}
    \text{minimize}   & \quad \|L\|_* + \lambda \|S\|_1\\
    \text{subject to} & \quad L + S = M 
  \end{array}
\]
can be performed using off-the-shelf tools such as interior point methods \cite{Boyd-cvx}. This was suggested for rank minimization in \cite{FazelM2003-ACC,Recht2008-SR} and for low-rank and sparse decomposition \cite{Venkat-09} (see also \cite{Liu2009}). However, despite their superior convergence rates, interior point methods are typically limited to small problems, say $n < 100$, due to the $O(n^6)$ complexity of computing a step direction.  

The limited scalability of interior point methods has inspired a
recent flurry of work on first-order methods. Exploiting an analogy
with iterative thresholding algorithms for $\ell_1$-minimization
\cite{Yin2008,Yin2008-SJIS}, Cai et.\ al.\ developed an algorithm that
performs nuclear-norm minimization by repeatedly shrinking the {\em
  singular values} of an appropriate matrix, essentially reducing the
complexity of each iteration to the cost of an SVD
\cite{Cai2008}. However, for our low-rank and sparse decomposition
problem, this form of iterative thresholding converges slowly,
requiring up to $10^4$ iterations. Ma et.\ al.\
\cite{Goldfarb2009,Ma2009} suggest improving convergence using
continuation techniques, and also demonstrate how Bregman iterations
\cite{Osher2005} can be applied to nuclear norm minimization.

The convergence of iterative thresholding has also been greatly
improved using ideas from Nesterov's optimal first-order algorithm for
smooth minimization \cite{Nesterov1983-SMD}, which was extended to
non-smooth optimization in \cite{Nesterov2005,Beck2009}, and applied
to $\ell_1$-minimization in
\cite{Nesterov2007,Beck2009,Becker2009}. Based on \cite{Beck2009}, Toh
et.\ al.\ developed a proximal gradient algorithm for matrix
completion which they termed {\em Accelerated Proximal Gradient
  (APG)}. A very similar APG algorithm was suggested for low-rank and
sparse decomposition in \cite{Lin2009-APG}. That algorithm inherits
the optimal $O(1/k^2)$ convergence rate for this class of
problems. Empirical evidence suggests that these algorithms can solve
the convex PCP problem at least 50 times faster than straightforward
iterative thresholding (for more details and comparisons, see
\cite{Lin2009-APG}).

However, despite its good convergence guarantees, the practical
performance of APG depends strongly on the design of good continuation
schemes. Generic continuation does not guarantee good accuracy and
convergence across a wide range of problem settings.\footnote{In our
  experience, the optimal choice may depend on the relative magnitudes
  of the $L$ and $S$ terms and the sparsity of the corruption.} In
this paper, we have chosen to instead solve the convex PCP problem
\eqref{eq:sdp} using an augmented Lagrange multiplier (ALM) algorithm
introduced in \cite{Lin2009-ALM,Yuan2009}. In our experience, ALM
achieves much higher accuracy than APG, in fewer iterations. It works
stably across a wide range of problem settings with no tuning of
parameters. Moreover we observe an appealing (empirical) property: the
rank of the iterates often remains bounded by $\mathrm{rank}(L_0)$
throughout the optimization, allowing them to be computed especially
efficiently. APG, on the other hand, does not have this property.

The ALM method operates on the {\em augmented Lagrangian} 
\begin{equation}
l(L,S,Y) = \| L \|_* + \lambda \| S \|_1 + \langle Y, M-L-S\rangle +
\dfrac{\mu}{2}\|M-L-S\|_F^2.
\end{equation}
A generic Lagrange multiplier algorithm \cite{Bertsekas-LM} would
solve PCP by repeatedly setting $(L_k,S_k) = \arg \min_{L,S}
l(L,S,Y_k)$, and then updating the Lagrange multiplier matrix via
$Y_{k+1} = Y_{k} + \mu(M - L_k- S_k)$.

For our low-rank and sparse decomposition problem, we can avoid having
to solve a sequence of convex programs by recognizing that $\min_L
l(L,S,Y)$ and $\min_S l(L,S,Y)$ both have very simple and efficient
solutions. Let $\mathcal{S}_\tau : \mathbb{R} \to \mathbb{R}$ denote
the shrinkage operator $\mathcal{S}_\tau[x] = \sgn(x) \max( |x|-\tau,
0 )$, and extend it to matrices by applying it to each element. It is
easy to show that
\begin{equation}
\arg \min\limits_S l(L,S,Y) = \mathcal{S}_{\lambda \mu}(M-L+\mu^{-1} Y).
\end{equation}
Similarly, for matrices $X$, let $\mathcal{D}_\tau(X)$ denote the singular value thresholding operator given by $\mathcal{D}_\tau(X) = U \mathcal{S}_\tau(\Sigma) V^*$, where $X = U \Sigma V^*$ is any singular value decomposition. It is not difficult to show that 
\begin{equation}
\arg \min\limits_L l(L,S,Y) = \mathcal{D}_\mu( M-S-\mu^{-1} Y ).
\end{equation}
Thus, a more practical strategy is to first minimize $l$ with respect to $L$ (fixing $S$), then minimize $l$ with respect to $S$ (fixing $L$), and then finally update the Lagrange multiplier matrix $Y$ based on the residual $M-L-S$, a strategy that is summarized as Algorithm \ref{alg:RPCA_ALM} below. 
\begin{algorithm}[th]
\caption{\bf (\name\, by Alternating Directions \cite{Lin2009-ALM,Yuan2009})}
\begin{algorithmic}[1]
\STATE {\bf initialize:} $S_0 = Y_0 = 0, \mu > 0$.
\WHILE{not converged} 
\STATE compute $L_{k+1} = \mathcal{D}_\mu( M-S_k-\mu^{-1} Y_k)$;
\STATE compute $S_{k+1} = \mathcal{S}_{\lambda \mu}(M-L_{k+1}+\mu^{-1} Y_k)$;
\STATE compute $Y_{k+1} = Y_k + \mu (M - L_{k+1} - S_{k+1})$;
\ENDWHILE
\STATE {\bf output:} $L,S$.
\end{algorithmic}
\label{alg:RPCA_ALM}
\end{algorithm}

Algorithm 1 is a special case of a more general class of augmented
Lagrange multiplier algorithms known as {\em alternating directions}
methods \cite{Yuan2009}. The convergence of these algorithms has been
well-studied (see e.g.\ \cite{Lions1979,Kontogiorgis1989} and the many
references therein, as well as discussion in
\cite{Lin2009-ALM,Yuan2009}). Algorithm \ref{alg:RPCA_ALM} performs excellently on a
wide range of problems: as we saw in Section 3, relatively small
numbers of iterations suffice to achieve good relative accuracy. The
dominant cost of each iteration is computing $L_{k+1}$ via singular
value thresholding. This requires us to compute those singular vectors
of $M-S_k-\mu^{-1} Y_k$ whose corresponding singular values exceed the
threshold $\mu$. Empirically, we have observed that the number of such
large singular values is often bounded by $\mathrm{rank}(L_0)$,
allowing the next iterate to be computed efficiently via a partial
SVD.\footnote{Further performance gains might be possible by replacing
  this partial SVD with an approximate SVD, as suggested in
  \cite{Goldfarb2009} for nuclear norm minimization.} The most
important implementation details for this algorithm are the choice of
$\mu$ and the stopping criterion. In this work, we simply choose $\mu
= n_1 n_2 / 4 \| M \|_1$, as suggested in \cite{Yuan2009}. We
terminate the algorithm when $\| M - L - S \|_F \le \delta \| M \|_F$,
with $\delta = 10^{-7}$.

Very similar ideas can be used to develop simple and effective
augmented Lagrange multiplier algorithms for matrix completion
\cite{Lin2009-ALM}, and for the robust matrix completion problem
\eqref{eq:sdp2} discussed in Section \ref{sec:matrix-completion}, with
similarly good performance. In the preceding section, all simulations
and experiments are therefore conducted using ALM-based
algorithms. For a more thorough discussion, implementation details and
comparisons with other algorithms, please see
\cite{Lin2009-ALM,Yuan2009}.

\section{Discussion}
\label{sec:discussion}

This paper delivers some rather surprising news: one can disentangle
the low-rank and sparse components exactly by convex programming, and
this provably works under very broad conditions that are much broader
than those provided by the best known results.  Further, our analysis
has revealed rather close relationships between matrix completion and
matrix recovery (from sparse errors) and our results even generalize
to the case when there are both incomplete and corrupted entries
(i.e. Theorem \ref{teo:MCRobust}).  In addition, \name\, 
does not have any free parameter and can be solved by simple
optimization algorithms with remarkable efficiency and accuracy. More
importantly, our results may point to a very wide spectrum of new
theoretical and algorithmic issues together with new practical
applications that can now be studied systematically.

Our study so far is limited to the low-rank component being exactly
low-rank, and the sparse component being exactly sparse. It would be
interesting to investigate when either or both these assumptions are
relaxed. One way to think of this is via the new observation model $M
= L_0 + S_0 + N_0$, where $N_0$ is a dense, small perturbation
accounting for the fact that the low-rank component is only
approximately low-rank and that small errors can be added to all the
entries (in some sense, this model unifies the classical PCA and the
robust PCA by combining both sparse gross errors and dense small
noise). The ideas developed in \cite{Candes2009-PIEEE} in connection
with the stability of matrix completion under small perturbations may
be useful here. Even more generally, the problems of sparse signal
recovery, low-rank matrix completion, classical PCA, and robust PCA
can all be considered as special cases of a general measurement model
of the form
\[
M = \mathcal{A}(L_0) + \mathcal{B}(S_0) + \mathcal{C}(N_0),
\]
where $ \mathcal{A}, \mathcal{B}, \mathcal{C}$ are known linear maps.
An ambitious goal might be to understand exactly under what
conditions, one can effectively retrieve or decompose $L_0$ and $S_0$
from such noisy linear measurements via convex programming.

The remarkable ability of convex optimizations in recovering low-rank
matrices and sparse signals in high-dimensional spaces suggest that
they will be a powerful tool for processing massive data sets that
arise in  image/video processing, web data analysis, and
bioinformatics. Such data are often of millions or even billions of
dimensions so the computational and memory cost can be far beyond that of a
typical PC. Thus, one important direction for future investigation is
to develop algorithms that have even better scalability, and can be
easily implemented on the emerging parallel and distributed computing
infrastructures.

\section{Appendix}
\label{sec:appendix}

\subsection{Equivalence of sampling models}

We begin by arguing that a recovery result under the Bernoulli model
automatically implies a corresponding result for the uniform
model. Denote by $\P_{\text{Unif}(m)}$ and $\P_{\text{Ber}(p)}$
probabilities calculated under the uniform and Bernoulli models and
let ``Success'' be the event that the algorithm succeeds. We have
\begin{align*}
  \P_{\text{Ber}(p)}(\text{Success}) & = \sum_{k = 0}^{n^2} \P_{\text{Ber}(p)}(\text{Success} \, | \, |\Omega| = k) \P_{\text{Ber}(p)}(|\Omega| = k)\\
  & \le \sum_{k = 0}^{m-1} \P_{\text{Ber}(p)}(|\Omega| = k) + \sum_{k
    = m}^{n^2}
  \P_{\text{Unif}(k)}(\text{Success}) \P_{\text{Ber}(p)}(|\Omega| = k)\\
  & \le \P_{\text{Ber}(p)}(|\Omega| < m) +
  \P_{\text{Unif}(m)}(\text{Success}),
\end{align*}
where we have used the fact that for $k \ge m$, $\P_{\text{Unif}(k)}(
\text{Success} ) \le \P_{\text{Unif}(m)}( \text{Success} )$, and that
the conditional distribution of $\Omega$ given its cardinality is
uniform. Thus, 
\[
\P_{\text{Unif}(m)} (\text{Success}) \ge
\P_{\text{Ber}(p)}(\text{Success}) - \P_{\text{Ber}(p)}(|\Omega| <
m). 
\]
Take $p = m/n^2 + \epsilon$, where $\epsilon > 0$. The conclusion
follows from $\P_{\text{Ber}(p)}(|\Omega| < m) \le
e^{-\frac{\epsilon^2n^2}{2p}}$.  In the other direction, the same
reasoning gives
\begin{align*}
  \P_{\text{Ber}(p)}(\text{Success}) & \ge \sum_{k = 0}^{m} \P_{\text{Ber}(p)}(\text{Success} \, | \, |\Omega| = k) \P_{\text{Ber}(p)}(|\Omega| = k)\\
  & \ge \P_{\text{Unif}(m)}(\text{Success}) \sum_{k = 0}^{m} \P_{\text{Ber}(p)}(|\Omega| = k)\\
  & = \P_{\text{Unif}(m)}(\text{Success}) \P(|\Omega| \le m),
\end{align*}
and choosing $m$ such that $\P(|\Omega| > m)$ is exponentially small,
establishes the claim.

\subsection{Proof of Lemma \ref{teo:infty}}

The proof is essentially an application of Bernstein's inequality,
which states that for a sum of uniformly bounded independent random
variables with $|Y_k - \E Y_k| < c$,
\begin{equation}
 \label{eq:bern2}
 \P\Bigl(\sum_{k = 1}^n (Y_k - \E Y_k) > t\Bigr) \le 2 \exp\Bigl(-\frac{t^2}{2\sigma^2 + 2ct/3}\Bigr),
\end{equation}
where $\sigma^2$ is the sum of the variances, $\sigma^2 \equiv \sum_{k
 = 1}^n \text{Var}(Y_k)$. 

Define $\Omega_0$ via $\Omega_0 = \{(i,j) : \delta_{ij} = 1\}$ where
$\{\delta_{ij}\}$ is an independent sequence of Bernoulli variables
with parameter $\rho_0$.  With this notation, $Z' = Z - \rho_0^{-1} \PT
\POzero Z$ is given by 
\[
Z' = \sum_{ij} (1 - \rho_0^{-1} \delta_{ij}) Z_{ij} \PT(e_ie_j^*)
\]
so that $Z'_{i_0j_0}$ is a sum of independent random variables, 
\[
Z'_{i_0j_0} = \sum_{ij} Y_{ij}, \qquad Y_{ij} = (1 - \rho_0^{-1}
\delta_{ij}) Z_{ij} \< \PT(e_ie_j^*), e_{i_0} e_{j_0}^*\>.
\]
We have 
\begin{align*}
\sum_{ij} \text{Var}(Y_{ij}) & = (1-\rho_0) \rho_0^{-1}  \sum_{ij}
|Z_{ij}|^2 |\< \PT(e_ie_j^*), e_{i_0} e_{j_0}^*\>|^2\\
& \le   (1-\rho_0) \rho_0^{-1}  \|Z\|_\infty^2 \sum_{ij} |\< e_ie_j^*, \PT(e_{i_0} e_{j_0}^*)\>|^2\\
& =  (1-\rho_0) \rho_0^{-1}  \|Z\|_\infty^2 \|\PT(e_{i_0} e_{j_0}^*)\|_F^2\\
 & \le  (1-\rho_0) \rho_0^{-1}  \|Z\|_\infty^2 \frac{2\mu r}{n},
\end{align*}
where the last inequality holds because of \eqref{eq:PTeiej}. Also, it
follows from \eqref{eq:PU} that $|\< \PT(e_ie_j^*), e_{i_0}
e_{j_0}^*\>| \le \|\PT(e_ie_j^*)\|_F \|\PT(e_{i_0} e_{j_0}^*)\|_F \le
2\mu r/n$ so that $|Y_{ij}| \le \rho_0^{-1} \|Z\|_\infty \mu
r/n$. Then Bernstein's inequality gives
\[
\P(|Z'_{ij}| > \epsilon \|Z\|_\infty) \le 2 \exp\Bigl(-\frac{3}{16} \,
\frac{\epsilon^2 n \rho_0}{\mu r}\Bigr).
\]
If $\rho_0$ is as in Lemma \ref{teo:infty}, the union bound proves the
claim. 

\subsection{Proof of Theorem \ref{teo:MCRobust}}

\newcommand{\PG}{\mathcal{P}_\Gamma}
\newcommand{\PGp}{\mathcal{P}_{\Gamma^\perp}}

This section presents a proof of Theorem \ref{teo:MCRobust}, which
resembles that of Theorem \ref{teo:main}. Here and below, $S'_0 =
\POobs S_0$ so that the available data are of the form $Y = \POobs L_0
+ S'_0$. We make three observations.
\begin{itemize}
\item If PCP correctly recovers $L_0$ from the input data $\POobs L_0
  + S'_0$ (note that this means that $\hat L = L_0$ and $\hat S =
  S'_0$), then it must correctly recover $L_0$ from $\POobs L_0 +
  S''_0$, where $S''_0$ is a trimmed version of $S'_0$. The proof is
  identical to that of our elimination result, namely, Theorem
  \ref{teo:elimination}. The derandomization argument then applies and
  it suffices to consider the case where the signs of $S'_0$ are
  i.i.d.~symmetric Bernoulli variables.

\item It is of course sufficient to prove the theorem when each entry
  in $\Obs$ is revealed with probability $p_0 := 0.1$, i.e.~when $\Obs
  \sim \text{Ber}(p_0)$.

\item We establish the theorem in the case where $n_1 = n_2 = n$ as
  slight modifications would give the general case.
\end{itemize}

Further, there are now three index sets of interest:
\begin{itemize}
\item $\Obs$ are those locations where data are available.
\item $\Gamma \subset \Obs$ are those locations where data are
  available and clean; that is, $\PG Y = \PG L_0$.
\item $\Omega = \Obs \setminus \Gamma$ are those locations where data
  are available but totally unreliable.
\end{itemize}
The matrix $S'_0$ is thus supported on $\Omega$. If $\Obs \sim
\text{Ber}(p_0)$, then by definition, $\Omega \sim \text{Ber}(p_0
\tau)$.

\paragraph{Dual certification.} We begin with two lemmas concerning
dual certification.
\begin{lemma}
\label{teo:kkt2}
Assume $\|\PGp \PT\| < 1$.  Then $(L_0, S'_0)$ is the unique
solution if there is a pair $(W, F)$ obeying
\[
UV^* + W = \lambda(\sgn(S'_0) + F),  
\]
with $\PT W = 0 $, $\|W\| < 1$, $\PGp F = 0$ and $\|F\|_\infty < 1$.
\end{lemma}
The proof is about the same as that of Lemma \ref{teo:kkt}, and is
discussed in very brief terms. The idea is to consider a feasible
perturbation of the form $(L_0 + H_L, S'_0 - H_S)$ obeying $\POobs H_L
= \POobs H_S$, and show that this increases the objective functional
unless $H_L = H_S = 0$. Then a sequence of steps similar to that in
the proof of Lemma \ref{teo:kkt} establishes
\begin{equation}
\label{eq:kkt2}
\|L_0 + H_L\|_* + \lambda \|S'_0 - H_S\|_1 \ge \|L_0\|_* +
\lambda \|S'_0\|_1 + (1-\beta)(\|\PTp H_L\|_* + \lambda \|P_\Gamma
H_L\|_1), 
\end{equation}
where $\beta = \text{max}(\|W\|, \|F\|_\infty)$.  Finally, $\|\PTp
H_L\|_* + \lambda \|P_\Gamma H_L\|_1$ vanishes if and only if $H_L \in
\Gamma^\perp \cap T = \{0\}$.

\begin{lemma}
\label{teo:kktdg2}
Assume that for any matrix $M$, $\|\PT \PGp M\|_F \le n \|\PTp \PGp
M\|_F$ and take $\lambda > 4/n$.  Then $(L_0, S'_0)$ is the unique
solution if there is a pair $(W, F)$ obeying
\[
UV^* + W + \PT D = \lambda(\sgn(S'_0) + F),
\]
with $\PT W = 0 $, $\|W\| < 1/2$, $\PGp F = 0$ and $\|F\|_\infty < 1/2$,
and $\|\PT D\|_F \le n^{-2}$.
\end{lemma}
Note that $\|\PT \PGp M\|_F \le n \|\PTp \PGp M\|_F$ implies
$\Gamma^\perp \cap T = \{0\}$, or equivalently $\|\PGp \PT\| <
1$. Indeed if $M \in \Gamma^\perp \cap T$, $\PT \PGp M = M$ while
$\PTp \PGp M = 0$, and thus  $M = 0$. 

\begin{proof}
  It follows from \eqref{eq:kkt2} together with the same argument as
  in the proof of Lemma \ref{teo:kktdg2} that
\[
\|L_0 + H_L\|_* + \lambda \|S'_0 - H_S\|_1 \ge \|L_0\|_* +
\lambda \|S'_0\|_1 + \frac12 \Bigl(\|\PTp H_L\|_* + \lambda
\|\PG H_L\|_1\Bigr) - {1 \over n^2} \|\PT H_L\|_F.
\]
Observe now that 
\begin{align*}
  \|\PT H_L\|_F & \le \|\PT \PG H_L\|_F  + \|P_T\PGp H_L\|_F  \\
                          & \le \|\PT \PG H_L\|_F  + n\|\PTp \PGp  H_L\|_F\\
                          & \le \|\PT \PG H_L\|_F  + n(\|\PTp \PG H_L\|_F+\|\PTp H_L\|_F)\\
 & \le (n+1) \|\PG H_L\|_F  + n \|P_{T^{\perp}} H_L\|_F. 
\end{align*}
Using both $\|\PG H_L\|_F \le \|\PG H_L\|_1$ and $\|\PTp H_L\|_F \le
\|\PTp H_L\|_*$, we obtain 
\[
\|L_0 + H_L\|_* + \lambda \|S'_0 - H_S\|_1 \ge \|L_0\|_* +
\lambda \|S'_0\|_1 + \Bigl(\frac12 -\frac{1}{n}\Bigr)
\|\PTp H_L\|_* + \Bigl(\frac{\lambda}{2} -\frac{n+1}{n^2}\Bigr) \|\PG
H_L\|_1.
\]
The claim follows from $\Gamma^\perp \cap T = \{0\}$.
\end{proof}

\begin{lemma}
\label{teo:xiao}
Under the assumptions of Theorem \ref{teo:MCRobust}, the
assumption of Lemma \ref{teo:kktdg2} is satisfied with high
probability. That is, $\|\PT \PGp M\|_F \leq n\|\PTp \PGp M\|_F$ for
all $M$.
\end{lemma}
\begin{proof}
  Set $\rho_0 = p_0(1-\tau)$ and $M' = \PGp M$. Since $\Gamma \sim
  \text{Ber}(\rho_0)$, Theorem \ref{teo:rudelson} gives $\|\PT -
  \rho_0^{-1} \PT\PG\PT\|\leq 1/2$ with high probability. Further,
  because $\|\PG \PT M'\|_F = \|\PG \PTp M'\|_F$, we have
\[
\|\PG \PT M'\|_F \le \|\PTp M'\|_F.
\] 
In the other direction,
\begin{align*}
  \rho_0^{-1} \|\PG \PT M'\|_F^2 & = \rho_0^{-1} \<\PT M', \PT \PG \PT M')\\
  & = \<\PT M', \PT M'\> +  \<\PT M', (\rho_0^{-1} \PT\PG\PT-\PT) M')\\
  & \ge \|\PT M'\|_F^2 - \frac12 \|\PT M'\|_F^2 = \frac12 \|\PT M'\|_F^2. 
\end{align*}
In conclusion, $\|\PTp M'\|_F\geq \|\PG\PT M'\|_F\geq \frac{\rho_0}{2}
\|\PT M'\|_F$, and the claim follows since $\frac{\rho_0}{2} \ge
\frac{1}{n}$.
\end{proof}

\newcommand{\POone}{\mathcal{P}_{\Omega}} 

Thus far, our analysis shows that to establish our theorem, it suffices
to construct a pair $(Y^L, W^S)$ obeying
\begin{equation}
\label{eq:dual-certif2}
\begin{cases} 
  \|\PTp Y^L\| < 1/4,\\
  \|\PT Y^L - UV^*\|_F \le n^{-2},\\
  \PGp Y^L = 0,\\
  \|\PG Y^L\|_\infty < \lambda/4,\\
\end{cases}
\qquad \text{and} \qquad
\begin{cases} 
  \PT W^S = 0,\\
  \|W^S\| \le 1/4,\\
  \POone W^S = \lambda \sgn(S'_0),\\
\POobsp W^S = 0,\\
  \|\PG W^S\|_\infty \le \lambda/4. 
\end{cases}
\end{equation}
Indeed, by definition, $Y^L + W^S$ obeys
\[
Y^L + W^S = \lambda(\text{sgn}(S'_0) + F),
\]
where $F$ is as in Lemma \ref{teo:kktdg2}, and it can also be
expressed as
\[
Y^L + W^S = UV^* + W + \PT D, 
\]
where $W$ and $\PT D$ are as in this lemma as well. 

\newcommand{\PGj}{\mathcal{P}_{\Gamma_j}}

\paragraph{Construction of the dual certificate $Y^L$.} We use the
golfing scheme to construct $Y^L$. Think of $\Gamma \sim
\text{Ber}(\rho_0)$ with $\rho_0 = p_0(1-\tau)$ as $\cup_{1 \le j \le
  j_0} \Gamma_j$, where the sets $\Gamma_j \sim \text{Ber}(q)$ are
independent, and $q$ obeys $\rho_0 = 1 - (1-q)^{j_0}$. Here, we take
$j_0 = \lceil 3\log n \rceil$, and observe that $q \ge \rho_0/j_0$ as
before. Then starting with $Y_0 = 0$, inductively define
\[
Y_j = Y_{j-1} + q^{-1} \PGj \PT(UV^* - Y_{j-1}),
\]
and set
\begin{equation}
  \label{eq:YL}
  Y^L = Y_{j_0} = q^{-1} \sum_j \PGj Z_{j-1}, \quad Z_j = (P_T - q^{-1} \PT \PGj \PT) Z_{j-1}.  
\end{equation}
By construction, $\PGp Y^L = 0$.  Now just as in Section
\eqref{sec:WL}, because $q$ is sufficiently large, $\|Z_j\| \le e^{-j}
\|UV^*\|_\infty$ and $\|Z_j\|_F \le e^{-j} \sqrt{r}$, both inequality
holding with large probability. The proof is now identical to
that in \eqref{eq:WL}. First, the same steps show that
\[
\|\PTp Y^L\| \le C \sqrt{\frac{n\log n}{q}} \|UV^*\|_\infty = C'
\sqrt{\frac{\mu r (\log n)^2}{n \rho_0}}.
\]
Whenever $\rho_0 \ge C_0 \frac{\mu r (\log n)^2}{n}$ for a
sufficiently large value of the constant $C_0$ (which is possible
provided that $\rho_r$ in \eqref{eq:MCrobust} is sufficiently small),
this terms obeys $\|\PTp Y^L\| \le 1/4$ as required.  Second,
\[
\|\PT Y^L - UV^*\|_F = \|Z_{j_0}\|_F \le e^{-3\log n} \sqrt{r} \le
n^{-2}.
\]
And third, the same steps give
\[
\|Y^L\|_\infty \le q^{-1} \|UV^*\|_\infty \sum_j e^{-j} \le 3(1-e^{-1})
\sqrt{\frac{\mu r (\log n)^2}{\rho_0^2 n^2}}. 
\]
Now it suffices to bound the right-hand side by $\frac{\lambda}{4} =
\frac14 \sqrt{\frac{1-\tau}{n\rho_0}}$. This is automatic when $\rho_0
\ge C_0 \frac{\mu r (\log n)^2}{n}$ whenever $C_0$ is sufficiently
large and, thus, the situation is as before.  In conclusion, we have
established that $Y^L$ obeys \eqref{eq:dual-certif2} with high
probability.

\newcommand{\PTOp}{\mathcal{P}_{(T + \Obs^\perp)^\perp}}
\newcommand{\PTO}{\mathcal{P}_{(T + \Obs^\perp)}}
\newcommand{\cQ}{\mathcal{Q}}

\paragraph{Construction of the dual certificate $W^S$.} We first
establish that with  high probability, 
\begin{equation}
  \label{eq:PTPO2}
  \|\PT \PO\| \le  \sqrt{\tau' p_0}, \quad \tau' = \tau + \tau_0,
\end{equation}
where $\tau_0(\tau)$ is a continuous function of $\tau$ approaching
zero when $\tau$ approaches zero. In other words, the parameter
$\tau'$ may become arbitrary small constant by selecting $\tau$ small
enough.  This claim is a straight application of Corollary
\ref{teo:POPT}. We also have
\begin{equation}
  \label{eq:PTPO3}
  \|\POone \PTO \POone \| \leq 2\tau'.
\end{equation}
with high probability. This second claim uses the identity
\[
\POone \PTO \POone = \POone \PT 
(\PT\POobs\PT)^{-1}\PT\POone. 
\]
This is well defined since the restriction of $\PT \POobs \PT$ to $T$
is invertible. Indeed, Theorem \ref{teo:rudelson} gives $\PT \POobs
\PT \ge \frac{p_0}{2} \PT$ and, therefore, $\|(\PT\POobs\PT)^{-1}\|
\le 2 p_0^{-1}$. Hence, 
\[
\|\POone \PTO \POone\| \le 2p_0^{-1} \|\POone \PT\|^2, 
\]
and \eqref{eq:PTPO3} follows from \eqref{eq:PTPO2}. 

Setting $E = \sgn(S'_0)$, this allows to define $W^S$ via
\begin{align*}
W^S & = \lambda (\OpId - \PTO) (\POone - \POone \PTO \POone)^{-1} E\\
    & := (\OpId - \PTO) (W_0^S +  W_1^S),
\end{align*}
where $W_0^S = \lambda E$, and $W_1^S = \mathcal{R}E$ with
$\mathcal{R} = \sum_{k \ge 1} (\POone \PTO \POone)^k$.  The operator
$\mathcal{R}$ is self-adjoint and obeys $\|\mathcal{R}\| \le
\frac{2\tau'}{1-2\tau'}$ with high probability. By construction, $\PT
W^S = \POobsp W^S = 0$ and $\PO W^S = \lambda \sgn(S'_0)$. It remains
to check that both events $\|W^S\| \le 1/4$ and $\|\PG W^S\|_\infty
\le \lambda/4$ hold with high probability.

{\em Control of $\|W^S\|$}. For the first term, we have $\|(\OpId -
\PTO) W_0^S\| \le \|W_0^S\| = \lambda \|E\|$. Because the entries of
$E$ are i.i.d.~and take the value $\pm 1$ each with probability $p_0
\tau/2$, and the value $0$ with probability $1-p_0 \tau$, standard
arguments give
\[
\|E\| \le 4\sqrt{n p_0 (\tau+\tau_0)}
\]
with large probability. Since $\lambda = 1/\sqrt{p_0 n}$,
$\|W_0^S\| \le 4 \sqrt{\tau + \tau_0} < 1/8$ with high probability,
provided $\tau$ is small enough.

\newcommand{\cR}{\mathcal{R}}
For the second term, $\|(\OpId - \PTO) W_1^S\| \le \lambda \|\mathcal{R}
E\|$, and the same covering argument as before gives 
\[
\P(\lambda \|\mathcal{R}(E)\| > t) \le 2 \times 6^{2n} \, \exp\Bigl(- \frac{
  t^2}{2 \lambda^2 \sigma^2}\Bigr) + \P(\|\cR\| \ge \sigma).
\]
Since $\lambda = 1/\sqrt{np_0}$ this shows that $\|W^S\| \le 1/4$ with
high probability, since one can always choose $\sigma$, or
equivalently $\tau' = \tau + \tau_0$, sufficiently small.

{\em Control of $\|\PG W^S\|_\infty$.} For $(i,j) \in \Gamma$, we have 
\[
W^S_{ij} = \<e_i e_j^*, W^S\> = \lambda \<X(i,j), E\>,
\]
where 
\[
X(i,j) = (\POone - \POone \PTO \POone)^{-1} \POone \PTOp e_i e_j^*. 
\]
The same strategy as before gives
\[
\P\Bigl(\sup_{(i,j) \in G} |W^S_{ij}| > \frac{\lambda}{4} \Bigr) \le 2n^2
\exp\Bigl(-\frac{1}{8\sigma^2}\Bigr) + \P\Bigl(\sup_{(i,j) \in G}
\|X(i,j)\|_F > \sigma\Bigr).
\]
It remains to control the Frobenius norm of $X(i,j)$. To do this, we
use the identity
\[
\POone \PTOp e_i e_j^* = \POone \PT  (\PT \POobs \PT)^{-1} \PT e_ie_j^*,
\]
which gives
\[
\|\POone \PTOp e_i e_j^*\|_F \le \sqrt{\frac{4 \tau'}{p_0}} \|\PT e_i
e_j^*\|_F \le \sqrt{\frac{8 \mu r \tau'}{np_0}}
\]
with high probability. This follows from the fact that $\|(\PT \POobs
\PT)^{-1}\| \le 2p_0^{-1}$ and $\|\PO \PT\| \le \sqrt{p_0 \tau'}$ as
we have already seen. Since we also have $\|(\POone - \POone \PTO
\POone)^{-1}\| \le \frac{1}{1-2\tau'}$ with high probability,
\[
\sup_{(i,j) \in \Gamma} \|X(i,j)\|_F \le \frac{1}{1-2\tau'}
\sqrt{\frac{8 \mu r \tau'}{np_0}}. 
\]
This shows that $\|\PG W^S\|_\infty \le \lambda/4$ if $\tau'$, or
equivalently $\tau$, is sufficiently small.

\subsection*{Acknowledgements}

E. C. is supported by ONR grants N00014-09-1-0469 and N00014-08-1-0749
and by the Waterman Award from NSF. Y. M. is partially supported by
the grants NSF IIS 08-49292, NSF ECCS 07-01676, and ONR
N00014-09-1-0230. E. C. would like to thank Deanna Needell for
comments on an earlier version of this manuscript.  We would also like
to thank Zhouchen Lin (MSRA) for his help with the ALM algorithm, and
Hossein Mobahi (UIUC) for his help with some of the simulations.

\small 
\bibliographystyle{plain}
\bibliography{RobustPCA}

\end{document}